\documentclass[a4paper]{article}
\usepackage{RR}
\usepackage{hyperref}

\usepackage{latexsym, stmaryrd, amssymb, amsbsy, amsmath} 
\usepackage{url}

\usepackage{comment}
\usepackage{ifthen}
\newboolean{extended-abstract}
\setboolean{extended-abstract}{false}
\ifthenelse{\boolean{extended-abstract}}%
{\excludecomment{RR}\includecomment{ABS}}%
{\includecomment{RR}\excludecomment{ABS}}
\newcommand{\remarque}[1]{}

\usepackage{pst-node, pst-tree}

\usepackage{float}
\floatstyle{boxed}
\newfloat{mafigure}{thp}{pipo}
\floatname{mafigure}{Figure}

\usepackage{theorem}
\theoremstyle{plain}
\theoremheaderfont{\rmfamily\bfseries}
\theorembodyfont{\itshape}
\newtheorem{theorem}{Theorem}
\newtheorem{definition}{Definition}{\bfseries}{\itshape}
\newtheorem{lemma}{Lemma}{\bfseries}{\itshape}
{\rmfamily}{\itshape}
\newtheorem{corollary}{Corollary}{\bfseries}{\itshape}
\theorembodyfont{\rmfamily}
\theoremheaderfont{\itshape}
\newtheorem{example}{Example}{\itshape}{\rmfamily}
{\itshape}{\rmfamily}
\newcommand{\qed}{\hfill $\Box$}

\newcommand{\finex}{\hfill $\Diamond$}
\newenvironment{proof}{\noindent\emph{Proof.}}{}

\newcommand{\R}{\mathcal{R}}
\newcommand{\T}{\mathcal{T}}
\renewcommand{\H}{\mathcal{H}}
\newcommand{\M}{\mathcal{M}}

\newcommand{\X}{\mathcal{X}}
\newcommand{\N}{\mathcal{N}}
\newcommand{\F}{\Sigma}

\newcommand{\A}{\mathcal{A}}

\newcommand{\C}{\mathcal{C}}
\newcommand{\G}{\mathcal{G}}
\newcommand{\var}{\mathit{var}}

\newcommand{\rootp}{\Lambda}
\newcommand{\ptrs}[2]{{#1}{/}{#2}}

\newcommand{\final}{\mathsf{f}}
\newcommand{\init}{\mathit{i}}

\newcommand{\pre}{\mathit{pre}}
\newcommand{\post}{\mathit{post}}

\newcommand{\dom}{\mathit{dom}}

\newcommand{\XACU}{\textsf{XACU}}
\newcommand{\INS}{\mathsf{INS}}
\newcommand{\REN}{\mathsf{REN}}
\newcommand{\RPL}{\mathsf{RPL}}
\newcommand{\DEL}{\mathsf{DEL}}
\newcommand{\XPath}{\mathsf{XPath}}
\newcommand{\maif}{\parallel}
\newcommand{\ann}{\mathit{ann}}

\def\frew#1#2#3#4#5#6#7#8{
\setbox0=\hbox{$#6 #7 #1 #8$}%
\setbox1=\hbox{$#6 #7 #2 #8$}%
\ifdim \wd0>\wd1 \rlap{\rlap{\hbox to \wd0{#5}}%
                            {\hbox to\wd0{\hfil\lower #3\box1\relax\hfil}}}{\raise #4\box0}%
\else \rlap{\rlap{\hbox to \wd1{#5}}{\hbox to\wd1{\hfil\raise #4\box0\relax\hfil}}}{\lower #3\box1}%
\fi
}

\def\fstep#1#2#3#4#5{\mathchoice{\frew{#1}{#2}{1.10ex}{1.20ex}{#5}{\scriptstyle}{#3}{#4}}%
                                {\frew{#1}{#2}{0.82ex}{1.20ex}{#5}{\scriptstyle}{#3}{#4}}%
                                {\frew{#1}{#2}{0.51ex}{0.82ex}{#5}{\scriptscriptstyle}{#3}{#4}}%
                                {\frew{#1}{#2}{0.51ex}{0.69ex}{#5}{\scriptscriptstyle}{#3}{#4}}}

\newcommand{\lrstep}[2]{\mathrel{\fstep{#1}{#2}{\;\>}{\>\>\;}{\rightarrowfill}}}
\newcommand{\rlstep}[2]{\mathrel{\fstep{#1}{#2}{\;\>\>}{\;\>}{\leftarrowfill}}}

\RRdate{July 2009}

\RRNo{7007}

\RRauthor{Florent Jacquemard and Michael Rusinowitch}

\authorhead{Jacquemard and Rusinowitch}

\RRtitle{Rewrite based Verification of XML Updates}
\RRetitle{Rewrite based Verification of XML Updates}
\titlehead{Rewrite based Verification of XML Updates}

\RRresume{%
We consider problems of access control for update of XML document.
In the context of XML programming, types can be viewed as hedge automata, and  static 
type checking amounts to verify that a program always converts valid source documents 
into also valid output documents. 
Given a set of update operations  we are particularly interested 
by checking safety properties such as preservation of document types 
along  any sequence of updates. 
We are also interested by the related policy consistency problem, that is  
detecting whether a sequence of authorized operations can simulate a forbidden one. 
We reduce these questions to type checking problems, 
solved by computing variants of hedge automata characterizing
the set of ancestors and descendants 
of the initial document type 
for the closure of parameterized rewrite rules.}

\RRabstract{%
We consider problems of access control for update of XML document.
In the context of XML programming, types can be viewed as hedge automata, and  static 
type checking amounts to verify that a program always converts valid source documents 
into also valid output documents. 
Given a set of update operations  we are particularly interested 
by checking safety properties such as preservation of document types 
along  any sequence of updates. 
We are also interested by the related policy consistency problem, that is  
detecting whether a sequence of authorized operations can simulate a forbidden one. 
We reduce these questions to type checking problems, 
solved by computing variants of hedge automata characterizing
the set of ancestors and descendants 
of the initial document type 
for the closure of parameterized rewrite rules.}

\RRmotcle{XML transformations, Typing, Software Verification, Tree Automata, Term Rewriting.}

\RRkeyword{XML transformations, Typing, Software Verification, Tree Automata, Term Rewriting.}

\RRprojets{DAHU and CASSIS}

\RRdomaine{4}
\RRtheme{Knowledge and Data Representation and Management}

\RCSaclay

\begin{document}

\makeRR

\section{Introduction}
XML has developed into the {de facto} standard for the exchange and manipulation
of data on the Web~\cite{AbiteboulBunemanSuciu}. 
XML documents are textual presentations of data stored in a tree structure, 
and are commonly represented as 
finite labeled unranked trees. 
In general, they are constrained by 
typing restrictions such as XML schemas expressing 
structural constraints on the organisation of the markups. 
Most of the typing formalisms currently used for XML are based on
finite tree automata. 
Several formalisms exist for the specification of transformation functions for XML documents,
e.g. for converting data from one source into a format suitable to a destination, 
for the automatic update of documents or 
the deletion of confidential data, e.g. for the enforcement of an access control policy
(wrapping or anonymization).
Among these formalisms, the W3C XQuery Update Facility~\cite{xqupdate}
defines some operations for document updates.

Applying transformation functions in the context of documents
following type constraints defined by schemas raises several compatibility problems.
Static Type Checking in the context of XML document processing amounts to verify at compile time 
that every XML document which is the result of a specified query or transformation 
of a  document with a valid input type produces an output document with a valid output type. 
Static Type Checking decidability is clearly dependant of the expressive power 
of the types and transformations that are employed. 
A standard approach to  XML type checking is 
forward (resp. backward) \emph{type inference}, 
that is the computation of 
an output (resp. input) XML type from given
input (resp. output) type and a tree transformation. 
Then the type checking itself 
can be reduced to  the verification of inclusion of the computed type in the given 
output or input type.

In this paper, motivated by XML access control problems, 
we consider document transformations that 
are arbitrary sequences of atomic update operations, and we address the problem of their type inference.
Since update operations, beside relabeling document nodes, 
\remarque{"updating XML content" = relabeling?}
can create and delete entire XML fragments,
modifying document's structure, 
it is not obvious to check whether they preserve the types of documents.

We propose a redefinition in term of rewrite rules (Section~\ref{sec:XACU})  
of the update operations of \XACU~\cite{FundulakiManeth07}, 
a formal model for specifying access control on XML data 
based on the W3C XQuery Update Facility draft~\cite{xqupdate}.
For these operations, and some proposed extensions,
we derive type inference algorithms that can also be 
employed to check access control policy local consistency
(i.e. to determine whether no sequence of allowed updates 
 starting from a given document can achieve 
 an explicitly forbidden update). 
Such situations may lead to serious security breaches 
and that are  challenging to detect according to~\cite{FundulakiManeth07}.
\remarque{+ provide control of application context}
Our results are obtained through the analysis of reachability sets of 
term rewriting systems for unranked trees, 
parametrized by hedge automata, and through the computation  
of an extension of hedge automata called context-free hedge automata. 
Therefore they may give more insight on these notions that have not been investigated before.

\medskip\noindent{\bf Related work:}
When considering general purpose transformation languages (e.g. XDuce, CDuce)
for writing transformations, typechecking is generally undecidable
and approximations must be applied.
In order to obtain exact algorithms, several approaches define conveniently
abstract formalisms for representing transformations.
Let us cite for instance TL (the transformation language)~\cite{ManethBPS05}
whose programs can be translated in macro tree transducers~\cite{MTT04},
and $k$-pebble tree transducers~\cite{MiloSuciuVianu03},
a powerful model defined so as to cover relevant fragments
of XSLT~\cite{XSLT} and other XML transformation languages.
Some restrictions on schema languages and on
top down tree transducers (on which transformations are based)
have also been studied~\cite{MartensNeven04}
in order to obtain PTIME type checking procedures.

The  results based on tree transducers are difficult to compare to ours. 
On one hand, we consider a small class of atomic update operations
whose expressiveness cannot be compared to general purpose transformation languages,
on the other hand, the application of updates is not restricted by strategies 
like e.g. top-down  transformations in~\cite{MartensNeven04}.
One can note that the works on typechecking
generally focus on the expressiveness of transformation languages,
and are restricted to XML types modeled as regular tree languages 
(languages of tree automata) or DTDs (a strict subclass of regular tree languages).
In our work we need to consider  XML types that generalize  regular tree languages 
and are recognized by context-free hedge automata~\cite{JR-rta2008}. 

The first access control model for XML was proposed by~\cite{DamianiVPS00} 
and was extended to secure updates in \cite{Lim03}. 
In \cite{Gabillon05}, the authors propose a solution to secure XUpdate queries.
Static analysis has been applied to XML Access Control in ~\cite{Murata06} to determine if a query expression 
is guaranteed not to access to elements that are forbidden by the policy. 
In ~\cite{FundulakiManeth07} the authors propose the XACU language. 
They study  policy consistency and show that it is undecidable in their setting. 
On the positive side~\cite{Bravo08} consider policies defined in term of annotated non recursive XML DTDs 
and give a polynomial algorithm for checking consistency. 
\remarque{c'est sur?}

Several recent works have considered the 
application of rewriting to reason about  
access control policies. 
These works do not adress  XML access control. 

\paragraph{Organization of the paper:} We introduce the needed 
formal background about terms, hedge automata and rewriting systems 
in Section \ref{sec:def}. Then we present XML update as parameterized  rewriting 
rules in Section \ref{sec:update}. Finally we give application to Access Control 
Policies in Section \ref{sec:acp}.

\section{Definitions}
\label{sec:def}

\subsection{Unranked Ordered Trees}

\paragraph{Terms and Hedges.}
We consider a finite alphabet $\F$ and an infinite set of variables $\X$.
The symbols of $\Sigma$ are generally denoted $a,b,c\ldots$ and the variables of
$\X$ $x$, $y$\ldots
The set $\H(\F, \X)$ of \emph{hedges} over $\F$ and $\X$ 
is the set of finite (possibly empty) sequences of terms
where the set of \emph{terms} over $\F$ and $\X$ is 
$\T(\F, \X) := \X \cup \bigl\{ a(h) \bigm| a \in \F, h \in \H(\F, \X) \bigr\}$.
The empty sequence is denoted $()$ and 
when $h$ is empty, the term $a(h)$ will be simply denoted by $a$.
We will sometimes consider a term as a hedge of length one, 
\textit{i.e.} consider that $\T(\F, \X) \subset \H(\F, \X)$.
A leaf of a hedge $(t_1\ldots t_n)$ is a leaf of one of the terms $t_1,..., t_n$.

The sets of ground terms (terms without variables) and ground hedges
are respectively denoted $\T(\F)$ and $\H(\F)$.
The set of variables occurring in a hedge $h \in \H(\Sigma, \X)$ is denoted $\var(h)$.
A hedge $h \in \H(\F, \X)$ is called \emph{linear} if every variable of $\X$ occurs at most once in $h$.

\noindent The root node of a term is denoted by $\rootp$.

\medskip\noindent{\bf Substitutions.}
A \emph{substitution} $\sigma$ is a mapping of finite domain from $\X$ into $\H(\F, \X)$.
The application of a substitution $\sigma$ to terms and hedges 
(written with postfix notation)
is defined recursively by
$x \sigma := \sigma(x)$ when $x \in \dom(\sigma)$,
$y \sigma := \sigma(x)$ when $y \in \X \setminus \dom(\sigma)$,
$(t_1 \ldots t_n) \sigma := (t_1\sigma \ldots t_n \sigma)$  for $n \geq 0$,
$f(h) \sigma := f(h\sigma)$.

\medskip\noindent{\bf Contexts.}
A \emph{context} is a hedge $u \in \H(\Sigma,\X)$
with a distinguished variable $x_u$ linear (with exactly one occurrence) in $u$.
The application of a context $u$ to a hedge $h \in \H(\Sigma,\X)$ 
is defined by $u[h] := u \{ x_u \mapsto h \}$:
it consists in inserting $h$ into an hedge in $u$ at the position of $x_u$.
Sometimes, we write $t[s]$ in order to emphasis that $s$ is a subterm (or subhedge) of $t$.


\subsection{Hedge Automata}
We consider two typing formalisms for XML documents,
defined as two classes of unranked tree automata.
The first class is the hedge-automata~\cite{Murata00}, denoted HA.
Most popular XML typing schemas like W3C XML Schemas 
or Relax NG are equivalent in expressiveness to HA.
The second and perhaps lesser known class is 
the context-free hedge automata, denoted CF-HA and introduced in~\cite{OhsakiST03}. 
CF-HA are strictly more expressive than HA and
we shall see that they are of interest for
the typing of certain update operations.

\begin{definition} \label{def:ha}
A \emph{hedge automaton} (resp. \emph{context-free hedge automaton}) 
is a tuple $\A = (\F, Q, Q^\final, \Delta)$ 
where  $\F$ is an finite unranked alphabet, 
$Q$ is a finite set of states disjoint from $\F$,
$Q^\final \subseteq Q$ is a set of final states, 
and $\Delta $ is a set of transitions of the form $a(L) \to q$ 
where $a \in \F$, $q\in Q$ and $L \subseteq Q^*$ is a regular word language 
(resp. a context-free word language). 
\end{definition}
When $\Sigma $ is clear from the context it is omitted in the tuple specifying~$\A$. 
We define the move relation between ground hedges 
in $h, h' \in \H(\F \cup Q)$ as follows: 
$h \lrstep{}{\A} h'$ iff there exists a context $u \in \H(\Sigma, \{ x_C\})$ 
and a transition $a(L) \rightarrow q \in \Delta$ such that 
$h = u[a(q_1\ldots q_n)]$, with $q_1\ldots q_n \in L$ and 
$h' = u[q]$. 
The relation  $\lrstep{*}{\A}$ is the transitive closure of $\lrstep{}{\A}$.

\paragraph{Collapsing Transitions.}
We consider the extension of HA and CF-HA with so called 
with \emph{collapsing transitions} which are special transitions of the form $L \to q$ 
where $L \subseteq Q^*$ is a CF language and $q$ is a state.
The move relation for the extended set of transitions
generalizes the above definition with the case
$u[q_1 \ldots q_n] \lrstep{}{\A} u[q]$
if $L \to q$ is a collapsing transition of $\A$ and $q_{1} \ldots q_{n} \in L$.
Note that we do not exclude the case $n = 0$ in this definition,
i.e. $L$ may contain the empty word in $L \to q$.
Collapsing transitions with a singleton language $L$ containing a length one word
(i.e. transitions of the form  $q \to q'$, where $q$ and $q'$ are states)
correspond to \emph{$\varepsilon$-transitions} for tree automata.

\paragraph{Languages.}
The language of a HA or CF-HA $\A$ in one of its states $q$,
denoted by $L(\A, q)$, is the set of ground hedges $h \in \H(\F)$ 
such that $h \lrstep{*}{\A} q$. 
We say sometimes that an hedge of $L(\A, q)$ has type $q$ (when $\A$ is clear from context).
A hedge is accepted by ${\A}$ if there exists $q \in  Q^\final$ such that $h \in L(\A, q)$. 
The language of $\A$, denoted by  $L(\A)$ is the set of hedge accepted by ${\A}$. 

Note that without collapsing transitions, all the hedges of $L(\A, q)$ are terms.
Indeed, by applying standard transitions of the form $a(L) \to a$, 
one can only reduce length-one hedges into states.
But collapsing transitions permit to reduce a ground hedge 
of length more than one into a single state.

The $\varepsilon$-transitions of the form $q \to q'$ do not
increase the expressiveness HA or CF-HA
(see~\cite{tata} for HA and the proof for CF-HA is similar).
But the situation is not the same in general for collapsing transitions:
collapsing transitions strictly extend HA in expressiveness,
and even collapsing transitions of the form $L \to q$
where the left member $L$ is a finite (hence regular) word language.
\begin{example}~\cite{JR-rta2008}.
The extended HA 
$\A= \bigl(\{ q, q_a, q_b, q_\final \},\{ g,a,b\}, \{q_\final \}, 
\{ 
           a \to q_a,
           b \to q_b,
           g(q) \to q_\final,
           q_a q_b \to q \}
\bigr)
\)
recognizes $\{g(a^n b^n) \mathbin{|} n \geq 1 \}$ which is not a HA language. 
\end{example}
However, collapsing transitions can be eliminated from CF-HA,
when restricting to the recognition of terms.

\begin{lemma}[\cite{JR-rta2008}] 
\label{lem:collapsing} \label{pr:collapse-HA} \label{pr:collapse-CFHA}
For every extended CF-HA over $\F$ with collapsing transitions $\A$, 
there exists a CF-HA $\A'$ without collapsing transitions such that 
$L(\A') = L(\A) \cap \T(\F)$.
\end{lemma}
%


\paragraph{Properties.}

It is known that for  both classes of HA and CF-HA
membership and emptiness problems are decidable in PTIME~\cite{Murata00,OhsakiST03}.
Moreover HA languages are closed under Boolean operations,
but CF-HA are not closed under intersection and complementation.
The intersection of a CF-HA language
and a HA language is a CF-HA language.
All these results are effective, with PTIME constructions of automata
of polynomial sizes for the closures under union and intersection.

\remarque{**RR** si pr. \S 3 sketched/suppr.}
We call a 
HA or CF-HA $\A = (\F, Q, Q^\final, \Delta)$ \emph{normalized} if
for every $a \in \F$ and every $q \in Q$, 
there is at most one transition rule $a(L_{a, q}) \to q$ in $\Delta$.
Every HA (resp. CF-HA) can be transformed into a normalized HA
(resp. CF-HA) in polynomial time
\label{page:normalization}
by replacing every two rules 
$a(L_1) \to q$ and $a(L_2) \to q$ by $a(L_1 \cup L_2) \to q$.

\subsection{Infinite Term Rewrite Systems} \label{sec:TRS}
We use term rewriting as a formalism for modeling XML update operations.
For this purpose, we propose a non-standard definition of term rewriting,
extending the classical one in two ways:
the application of rewrite rules is extended from ranked terms
to unranked terms and second,
the rules are parametrized by HA languages 
(i.e. each parametrized rule can represent an infinite number
 of unparametrized rules).

\medskip\noindent{\bf Term Rewriting Systems.} 
\label{sub:TRS}
A term rewriting system $\R$ over a finite unranked alphabet $\F$ (TRS)
is a set of \emph{rewrite rules} of the form $\ell \to r$ 
where $\ell \in \H(\Sigma,\X) \setminus \X$ and $r \in \H(\Sigma,\X)$;
$\ell$ and $r$ are respectively called left- and right-hand-side 
(\emph{lhs} and \emph{rhs}) of the rule.
Note that we do note assume the cardinality of $\R$ to be finite.

The rewrite relation $\lrstep{}{\R}$ of an TRS $\R$ is
the smallest binary relation on $\H(\F, \X)$ containing $\R$
and closed by application of substitutions and contexts.
In other words, $h \lrstep{}{\R} h'$ 
iff there exists a context $u$, 
a rule $\ell \to r \in \R$ and a substitution $\sigma$
such that $h = u[ \ell \sigma ]$ and $h' = u [ r\sigma ]$.
The reflexive and transitive closure of $\lrstep{}{\R}$ is denoted $\lrstep{*}{\R}$.


\begin{example} \label{ex:rewriting-cf}
\remarque{revoir c'est un example de \cite{JR-rta2008}}
With $\R = \{ g(x) \to x \}$, we have
$g(h) \lrstep{}{\R} h$ for all $h \in \H(\F, \X)$
(the term is reduced to the hedge $h$ of its arguments).
With $\R = \{ g(x) \to g(axb) \}$, $g(c) \lrstep{*}{\R} g(a^n c b^n)$
for every $n \geq 0$. 
\end{example}

\paragraph{Parametrized Term Rewriting Systems.} 
\label{sub:PTRS}
Let $\A = (\F, Q, Q^\final, \Delta)$ be a HA.
A term rewriting system over $\F$ and parametrized by $\A$ (PTRS)
(see \cite{Gilleron91}
is given by a finite set, denoted $\ptrs{\R}{\A}$, of 
rewrite rules $\ell \to r$ 
where 
$\ell \in \H(\Sigma,\X)$ and 
$r \in \H(\Sigma \uplus Q,\X)$ and symbols of $Q$ can only label leaves of $r$.
In this notation, $\A$ may be omitted when it is clear from context
or not necessary.
The rewrite relation $\lrstep{}{\ptrs{\R}{\A}}$ associated to a PTRS $\ptrs{\R}{\A}$ is
defined as the rewrite relation $\lrstep{}{\R[\A]}$ 
where the TRS $\R[\A]$ is the (possibly infinite) set of all rewrite rules
obtained from rules $\ell \to r$ in $\ptrs{\R}{\A}$
by replacing in $r$ every state $q \in Q$ by a ground hedge of $L(\A, q)$.
Several example of rewrite rules can be found in Figure~\ref{fig:XACU} below.


\paragraph{\bf Properties.}
Given a set $L \subseteq \H(\Sigma,\X)$ and a PTRS $\ptrs{\R}{\A}$, 
we denote by
$\post_{\ptrs{\R}{\A}}^*(L) = \{ h \in \H(\Sigma,\X) \mid \exists h' \in L, h' \lrstep{*}{\ptrs{\R}{\A}} h \}$
and 
$\pre_{\ptrs{\R}{\A}}^*(L) = \{ h \in \H(\Sigma,\X) \mid \exists h' \in L, h \lrstep{*}{\ptrs{\R}{\A}} h' \}$.

\noindent \emph{Ground reachability} is the problem to decide, given two hedges $h, h' \in \H(\F)$
and a PTRS $\ptrs{\R}{\A}$ whether $h \lrstep{*}{\ptrs{\R}{\A}} h'$.
Reachability problems for ground ranked tree rewriting have 
been investigated in e.g.~\cite{Gilleron91}. 
C. L\"oding ~\cite{Loding02} has obtained results in a more general setting 
where rules of type $L\rightarrow R$ specify the replacement of any element of a regular language $L$ 
by any element of a regular language $R$. Then ~\cite{LodingS07} has extended some of 
these works to unranked tree rewriting for the case of \emph{subtree and flat prefix rewriting} 
which is a combination of standard ground tree rewriting and prefix word rewriting on the ordered leaves of subtrees  
of height 1.

\noindent \emph{Typechecking} is the problem to decide, given 
two sets of terms $\tau_{in}$ and $\tau_{out}$ called input and output types
(generally presented as HA) and a PTRS $\ptrs{\R}{\A}$ whether 
$\post^*(\tau_{in}) \subseteq \tau_{out}$ or equivalently
$\tau_{in} \subseteq \pre^*(\tau_{out})$~\cite{MiloSuciuVianu03}.

Note that reachability is a special case of model checking,
when both $\tau_{in}$ and $\tau_{out}$ are singleton sets.
Hence typechecking is undecidable as soon as reachability is.

\noindent
One related problem, called \emph{type inference}, is,
given a of PTRS $\R / A$ and a HA or CF-HA language $L$,
to construct a HA or CF-HA recognizing $\post^*_\R(L)$ or $\pre^*_\R(L)$.


\section{Type Inference for Update Operations} \label{sec:update}

In this section, we address the problem of type inference
for arbitrary finite sequence of update operations.
More precisely, we propose a redefinition in term of PTRS rules (Section~\ref{sec:XACU})
of the update operations of \XACU~\cite{FundulakiManeth07} 
and some extensions.
Then, we show how to construct HA and CF-HA recognizing respectively
$\post^*_\R(L)$ and $\pre^*_\R(L)$ 
given a HA or CF-HA language $L$ and a PTRS $\R$ representing XACU operations 
(Sections~\ref{sec:post*})
or extended updates (Section~\ref{sec:post*XACU+}).

\remarque{** INTRO, DEL?}
The motivation for showing these results are twofold.
First, these constructions permit to address the problems of reachability and 
typechecking. \remarque{TODO.}
Second, they also permit the synthesis of missing input or output types.
Imagine that a PTRS $\R$ is given, as well
as an input type $\tau_{\mathit{in}}$, defined as an HA, but that the output type 
(for the application of rules of $\R$ to terms of $\tau_{\mathit{in}}$)
is not known.
The result of Theorem~\ref{th:post} ensures that we can build a CF-HA
recognizing $\post^*_{\R}(\tau_{\mathit{in}})$ 
and which can be use as a definition of a synthesized output type for $\R$.
Similarly, the result of Theorem~\ref{th:pre} can be used
to synthesis an input type, defined by the HA constructed for 
$\pre^*_{\R}(\tau_{\mathit{out}})$,
given an output type $\tau_{\mathit{out}}$ and a PTRS $\ptrs{\R}{\A}$.
\remarque{**}

\subsection{Update Operations} \label{sec:XACU}
Figure~\ref{fig:rules} displays PTRS rules corresponding 
to the rules of $\XACU$ as defined in~\cite{FundulakiManeth07} (in the first column)
and to some extensions (in the second column).
We call $\XACU$ the class of all PTRS containing rules of the kind presented in 
the first column of Figure~\ref{fig:rules},
and $\XACU+$ the class of all PTRS containing any rule 
presented in Figure~\ref{fig:rules}.

\begin{figure}
\[
\begin{array}{rclcl|crclcl}
\multicolumn{5}{c|}{\XACU} & & \multicolumn{5}{c}{\XACU+}\\
\hline
a(x) & \to & b(x)        & & \REN\\
a(x) & \to & a(p\, x)    & & \INS_\mathsf{first} & &
a(x) & \to & b(p\, x)    & & \INS'_\mathsf{first}\\
a(x) & \to & a(x\, p)    & & \INS_\mathsf{last} & &
a(x) & \to & b(x\, p)    & & \INS'_\mathsf{last}\\
a(xy) & \to & a(x\, p\, y) & & \INS_\mathsf{into} & &
   \\
a(x) & \to & p\, a(x)    & & \INS_\mathsf{left}\\
a(x) & \to & a(x)\, p    & & \INS_\mathsf{right}\\
a(x) & \to & p  & \quad & \RPL & &
a(x) & \to & p_1 \ldots p_n & \quad & \RPL'\\
a(x) & \to & ()     & & \DEL & &
a(x) & \to & x         & & \DEL_{\mathsf{s}}\\
\end{array}
\]
\caption{PTRS rules for XACU and extension}
\label{fig:rules} \label{fig:XACU}
\end{figure}

In this section we assume given an unranked alphabet $\F$ and
a HA $\A = (\F, Q, Q^\final, \Delta)$.
The rewrite rules are parametrized by
states $p$, $p_1$,..,$p_n$ of $\A$.

\medskip\noindent{\bf $\XACU$ rules.}
Let us first describe the update operations of $\XACU$ (see also~\cite{FundulakiManeth07}).
$\REN$ renames a node: it changes it label from $a$ into $b$.
Such a rule leaves the structure of the term unchanged.
$\INS_\mathsf{first}$ inserts a term of type $p$ at the first position
below a node labeled by $a$.
$\INS_\mathsf{last}$ inserts at the last position and
$\INS_\mathsf{into}$ at an arbitrary position below a node labeled by $a$.
$\INS_\mathsf{left}$ (resp. $\INS_\mathsf{right}$) insert a term of
type $p$ at the left (resp. right) sibling position to a node labeled by $a$.
$\DEL$ deletes a whole subterm whose root node is labeled by $a$
and $\RPL$ replaces such a subterm by a term of type $p$.

\begin{example} \label{ex:hospital}
The patient data in a hospital are stored in an XML 
document whose DTD type  can be recognized by an HA  $\A$ with rules: 
\[
\begin{array}{ccc}
\begin{array}[t]{rcl}
\mathsf{hospital}(p_\mathsf{p}^*) &\to & p_\mathsf{h}\\
\mathsf{patient}(p_\mathsf{n}\, p_{\mathsf{t}})  &\to & p_\mathsf{pa}\\
\mathsf{patient}(p_{\mathsf{n}}) & \to & p_{\mathsf{epa}}\\
\mathsf{treatment}(p_\mathsf{dr}\,p_\mathsf{dia}\,p_{\mathsf{da}}) &\to & p_\mathsf{t}
\end{array}
&
\begin{array}[t]{rcl}
\mathsf{name}(p_c^*) & \to & p_\mathsf{n} \\
\mathsf{drug}(p_c^*) &\to & p_\mathsf{dr} \\
\mathsf{diagnosis}(p_c^*) & \to & p_\mathsf{dia} \\
\mathsf{date}(p_c^*) &\to & p_{\mathsf{da}}
\end{array}
& 
\begin{array}[t]{rcl}
\mathsf{a} & \to & p_c\\
\mathsf{b} & \to & p_c\\
\mathsf{c} & \to & p_c\\
 \vdots &
\end{array}
\end{array}
\]

For instance we can use a $\DEL$ rule $\mathsf{patient}(x) \to ()$ for deleting 
a $\mathsf{patient}$,
and a $\INS_{\mathsf{last}}$ rule
$\mathsf{hospital}(x) \to \mathsf{hospital}(x\,p_{\mathsf{pa}})$ 
to insert a new $\mathsf{patient}$, 
at the last position below the root node $\mathsf{hospital}$. 
We can ensure that the patient newly added has
an empty $\mathsf{treatments}$ list (to be completed later)
using the rule $\mathsf{hospital}(x) \to \mathsf{hospital}(x\,p_{\mathsf{epa}})$.
The $\INS_{\mathsf{right}}$ rule 
$\mathsf{name}(x) \to \mathsf{name}(x)\,p_{\mathsf{t}}$
can be used to insert later a $\mathsf{treatment}$ next to the patient's $\mathsf{name}$.
\end{example}

\medskip\noindent{\bf Extended rules.}
In $\XACU+$ we introduce several extensions of the rules of $\XACU$.
\begin{RR}
We shall see in Section~\ref{sec:post*XACU+}
that the typing of these extended operations is different 
from the typing of the operation of $\XACU$:
while the type of terms obtained by $\XACU$ operations
can be described by HA, CF-HA must be used in order to describe
the type of terms obtained by $\XACU+$.
\end{RR}
%
A restriction of the insertion rules of $\XACU$ (the rules called $\INS_*$),
following the definitions in ~\cite{FundulakiManeth07},
is that the label of the node at the top of the lhs of the rules is left unchanged.
Only the rule $\REN$ permits to change the label of a node in a term,
while preserving the other nodes.
The rules $\INS'_*$ combine the application of the corresponding 
insert operation $\INS_*$ and of a node renaming $\REN$.
We will see in Section~\ref{sec:post*XACU+} that allowing such combinations 
has important consequences wrt type inference.


\noindent
The rule $\DEL_{\mathsf{s}}$ deletes a single node $n$ whose arguments 
inherit the position. It can be employed to build a user view as in ~\cite{FanChan04}. 
\remarque{TODO}
\begin{example}
Assume that some patients of the hospital of Example~\ref{ex:hospital}
are grouped into one category like in
$\mathsf{hospital}(\ldots \mathsf{priority}(p_{\mathsf{pa}}^*)\ldots)$,
and that we want to delete the category $\mathsf{priority}$ while keeping the patients information.
This can be done with the $\DEL_{\mathsf{s}}$ rule
$\mathsf{priority}(x) \to x$.
\end{example}

\noindent
Finally, with $\RPL'$ we slightly generalize the rule $\RPL$ by allowing 
a subterm whose root node is labeled by $a$ 
to be replaced
by a sequence of $n$ terms of respective types $p_1$,\ldots, $p_n$.

\noindent
Note that $\RPL$ and $\DEL$ are special cases of $\RPL$, with $n=1$ and $n = 0$ respectively.


\subsection{Forward Type Inference for $\XACU$ Rules} 
\label{sec:post*XACU} \label{sec:post*}
In this section and the following, we want to characterize the sets of terms which can 
be obtained, from terms of a given type,
by arbitrary application of updates operations as PTRS rules.
For this purpose, we shall study the recognizability (by HA and CF-HA),
of the forward closure ($\post^*$) of automata languages
under the above rewrite rules.

\begin{theorem} \label{th:post}
Given a HA $\A$ on $\Sigma$ and a PTRS $\ptrs{\R}{\A} \in \XACU$,
for all HA language $L$, $\post_{\ptrs{\R}{\A}}^*(L)$ 
is the language of an HA of size polynomial 
and which can be constructed in PTIME on 
the size of $\R$, $\A$ and an HA of language $L$. 
\end{theorem}
%
%
\begin{proof}
(sketch, see Appendix~\ref{app:post*} for a complete proof).
We consider a normalised HA $\A_L$ recognizing $L$ and add transitions (but no states) 
to the NFAs defining its horizontal languages in transitions $a(L_{a, q}) \to q$.
For instance, if $a(x) \to a(p\, x) \in \ptrs{\R}{\A}$ 
we add one transition $(i_{a,q}, p, i_{a,q})$ looping on the initial state
$i_{a,q}$ of the NFA for $L_{a, q}$.
If $a(x) \to a(x)\, p \in \ptrs{\R}{\A}$, 
and there exists a transition $(s, q, s')'$ in some NFA,
we add one transition $(s', p, s')$.
\qed
\end{proof}

Let us come back to our motivations.
A first consequence of Theorem~\ref{th:post}
concerns to the typechecking problem.
\begin{corollary} \label{cor:post*}
The typechecking is decidable in PTIME for $\XACU$.
\end{corollary}
\begin{proof}
Let $\tau_{\mathit{in}}$ and $\tau_{\mathit{out}}$
be two HA languages (resp. input and output types),
and let $\ptrs{\R}{\A}$ by a PTRS.
We want to know
whether $\post^*_{\ptrs{\R}{\A}}(\tau_{\mathit{in}}) \subseteq \tau_{\mathit{out}}$.
Following Theorem~\ref{th:post},
$\post^*_{\ptrs{\R}{\A}}(\tau_{\mathit{in}})$ is a HA language.
Hence 
$\post^*_{\ptrs{\R}{\A}}(\tau_{\mathit{in}}) \cap \overline{\tau_{\mathit{out}}}$
is a HA language, and testing its emptiness solves the problem.
\qed
\end{proof}

Regarding the problem of type synthesis, 
if we are given $\ptrs{\R}{\A}$ and 
an input type $\tau_{\mathit{in}}$, 
Theorem~\ref{th:post}
provides an output type presented as a HA.


\subsection{Forward and Backward Type Inference for $\XACU+$ Rules} 
\label{sec:post*XACU+}

Theorem~\ref{th:post} is no longer true for the rules of the extension $\XACU+$:
the examples below show that
the rules of $\XACU{+} \setminus \XACU$ do not preserve HA languages in general.
However, we prove in Theorem~\ref{th:postXACU+} that 
the rules of $\XACU+$ preserve the larger class of CF-HA language.

\begin{example}
Let $\F = \{ a,b,c,c'\}$ and let $\R$ be the finite TRS containing the two
$\INS'_\mathsf{first}$ and $\INS'_\mathsf{last}$ rules
${ c(x) \to c'(ax), c'(x) \to c(xb)}$.
We have $\post^*_\R\bigl(\{ c\}\bigr) \cap \H(\F) = \{ c(a^n b^n) \mid n \geq 0 \}$,
and this set is not a HA language. It follows that $\post^*_\R\bigl(\{ c \}\bigr)$
is not a HA language.\finex
\end{example}


\begin{example}
Let $\F = \{ a,b,c\}$, let $\R$ be the finite TRS with
one $\DEL_{\mathsf{s}}$ rule ${ c(x) \to x}$ and
let $L$ be the HA language containing exactly the terms $c(a c(a \ldots c \ldots b) b)$;
it is recognized by the HA with the set of transition rules 
$\bigl\{ a \to q_a, b \to q_b, c\bigl(\{ (), q_a\, q\, q_b\}\bigr) \to q \bigr\}$.
We have $\post^*_\R(L) \cap c\bigl( \{ a, b \}^*\bigr)   = \{ c(a^n b^n) \mid n \geq 0 \}$,
hence $\post^*_\R(L)$  is not a HA language.\finex
\end{example}



\begin{theorem} \label{th:postXACU+}
Given a HA $\A$ on $\Sigma$ and a PTRS $\ptrs{\R}{\A} \in \XACU+$,
for all CF-HA term language $L$, $\post_{\ptrs{\R}{\A}}^*(L)$ 
is the language of an CF-HA of size polynomial 
and which can be constructed in PTIME on 
the size of $\R$, $\A$ and an CF-HA recognizing $L$. 
\end{theorem}
%
%
\begin{proof}
(sketch, see Appendix~\ref{app:post*XACU+} for a complete proof).
We consider a normalised HA $\A_L$ recognizing $L$ and,
very roughly, we define new CFG $\G_{a,q}$ for the horizontal languages
as the union of CFG of transitions of $A_L$ with a new initial non-terminal $I'_{a,q}$
and new production rules according to $\ptrs{\R}{\A}$.
For instance, if $a(x) \to b(x) \in \ptrs{\R}{\A}$, 
we add a production rule $I'_{b,q} := I'_{a,q}$ and
for $a(x) \to b(p\, x)$, we add 
$I'_{b,q} := p I'_{a,q}$.
Moreover, we also add collapsing transitions
like $p_1 \ldots p_n \to q$ if $a(x) \to p_1 \ldots p_n \in \ptrs{\R}{\A}$.
\qed
\end{proof}

\begin{corollary}
The typechecking is decidable in PTIME for $\XACU+$.
\end{corollary}
\begin{proof}
The proof is the same as for Corollary~\ref{cor:post*},
because the intersection of a CF-HA and a HA language
is a CF-HA language 
(and there is an effective PTIME construction of an CF-HA of polynomial size)
and emptiness of CF-HA is decidable in PTIME.
\qed
\end{proof}


\label{sec:pre*}
\begin{theorem} \label{th:pre}
Given a HA $\A$ on $\Sigma$ and a PTRS $\ptrs{\R}{\A} \in \XACU+$,
for all HA language $L$, $\pre_{\ptrs{\R}{\A}}^*(L)$
is a HA the language.
\end{theorem}
%
%

Regarding the problem of type synthesis
for a $\ptrs{\R}{\A} \in \XACU+$,
if only an output type $\tau_{\mathit{out}}$ is given, 
then Theorem~\ref{th:postXACU+} provides an input type 
for $\ptrs{\R}{\A}$ presented as a HA, 
and if only an input type $\tau_{\mathit{in}}$ is given, 
then Theorem~\ref{th:postXACU+}
provides an output type presented as a CF-HA.
Unlike HA, CF-HA are not popular type schemas,
but HA solely do not permit to extend the results 
of Theorem~\ref{th:post} as shown by the above examples.


\section{Access Control Policies for Updates}
\label{sec:acp}
In this last section we study some models of Access Control Policies (ACP)
for the update operations defined in Section~\ref{sec:update},
and verification problems for these ACP.

\subsection{Term Rewrite Systems with Global Membership Constraints}
The ACP language $\XACU_\mathsf{annot}$ introduced in~\cite{FundulakiManeth07}
follows the approach of DTD with security annotations of~\cite{FanChan04}
to specify the read and write access authorizations for XML documents
in the presence of a DTD.
Annotated DTDs offer an elegant formalism for ACP specification,
which is especially convenient for developing techniques of type analysis.
However, it imposes the strong restriction that every document $t$ 
to which we want to apply an update operation (under the given ACP)
must comply to the DTD $D$ used for the ACP specification.

In our rewrite based formalism, this condition may be expressed 
by adding global constraints to the parametrized rewrite rules of Section~\ref{sec:TRS}.
These global constraints restrict the whole term to be rewritten 
(not only the redex) to belong to a given regular language.
Theorem~\ref{th:undec} below shows that, unfortunately, 
adding such constraints to ground rules (which are a very special kind of $\RPL$ rules)
makes the reachability undecidable.

Given a HA $\A = (\F, Q, Q^\final, \Delta)$,
a term rewriting system over $\F$, parametrized by $\A$ 
and with global constraints (PGTRS)
is given by a finite set, denoted $\ptrs{\R}{\A}$, of 
constrained rewrite rules 
$L :: \ell \to r $ 
where $\ell$ and $r$ satisfy the conditions of the rewrite rules of 
Section~\ref{sub:PTRS} and $L \subseteq \T(\Sigma)$ is a HA language.
A PGTRS is called \emph{uniform} if the language $L$ is the same for every rule.
The rewrite relation for PGTRS 
is defined as the restriction of the relation defined in Section~\ref{sub:PTRS} to ground terms:
for the application of a rule $L:: \ell \to r $ to a term $t$, 
we require that $t \in L$. 

\begin{theorem} \label{th:undec}
Reachability is undecidable for uniform PGTRS without variables and parameters.
\end{theorem}
The result can be contrasted with some decidability results on ground rewriting ~\cite{Gilleron91}.   
It is also a refinement of \cite{FundulakiManeth07}
where XPath queries are used filter out nodes where the updates apply.
%
As a corollary, reachability, 
hence inconsistency (see Section~\ref{sec:inconsistency}), 
are undecidable for $\XACU_{annot}$ ACP based on annotated recursive DTDs.


\subsection{$\XACU_2+$: Rewrite Rules with Context Control}
\label{sec:context}

The PTRS rewrite rules of Section~\ref{sec:update}
permit to define a minimal control for the application of the updates operations.
Indeed, all the lhs of rules have the form $a(x)$ 
(or $a(xy)$ for $\INS_{\mathsf{into}}$),
meaning that the application to such rules is restricted to nodes labeled
with $a$ (i.e. to nodes of DTD element type $a$ if the document conforms to a given fixed DTD).
%

For the rules with an hedge at rhs
(like $\INS_\mathsf{left}$, $\INS_\mathsf{right}$, $\RPL$, $\DEL$, $\DEL_{\mathsf{s}}$...)
we can extend this idea by furthermore constraining the label of the node 
at the parent node of the performed update.
The generalized rules are defined in Figure~\ref{fig:XACU2}.
\begin{example}
The $\DEL_2$ rule $\mathsf{hospital}(y\,\mathsf{patient}(x)\,z) \to \mathsf{hospital}(y\,z)$ can be used
to delete a $\mathsf{patient}$ only if it is located under a $\mathsf{hospital}$ node.
\end{example}

\begin{figure}
\[
\begin{array}{rclcl|crclcl}
\multicolumn{5}{c|}{\XACU_2} & & \multicolumn{5}{c}{\XACU_2+}\\
\hline
b(y\, a(x)\, z) & \to & b(y\, p\, a(x)\, z)    & & \INS_{2,\mathsf{left}}\\
b(y\, a(x)\, z) & \to & b(y\, a(x)\, p\, z)    & & \INS_{2,\mathsf{right}}\\
b(y\, a(x)\, z) & \to & b(y\, p\, z)  & & \RPL_2 & &
b(y\, a(x)\, z) & \to & b(y\, p_1 \ldots p_n\, z) & & \RPL_2'\\
b(y\, a(x)\, z) & \to & b(y\, z)     & & \DEL_2 & &
b(y\, a(x)\, z) & \to & b(y\, x\, z)         & & \DEL_{2,\mathsf{s}}\\
\end{array}
\]
\caption{PTRS rules for XACU with context control}
\label{fig:XACU2}
\end{figure}

This approach can be compared to the annotated DTD of~\cite{FanChan04}.
The security annotations of~\cite{FanChan04} are indeed mappings $\ann$ from 
pairs of DTD elements types $(b,a)$ into values of $Y$, $N$ or $[q]$
(for resp. read access allowed, denied or conditionally allowed, where $q$ is an XPath qualifier).
An annotation $\ann(b,a) = Y$ or $N$ or $[q]$ indicates
that the $a$ children of $b$ elements (in an instantiation of the given DTD $D$) 
are accessible, inaccessible or conditionally accessible respectively.
This approach is limited to the case of unambiguous DTDs,
where the element type $a$ can have at most one element $b$ as parent.

Let us call $\XACU_2+$ the class of all PTRS containing rules of $\XACU+$ or
rules of the kind described in Figure~\ref{fig:XACU2}.
The construction of Theorem~\ref{th:pre} 
for backward type inference 
can be straightforwardly extended from $\XACU+$ to $\XACU_2+$.

\begin{theorem} \label{th:pre2}
Given a HA $\A$ on $\Sigma$ and a PTRS $\ptrs{\R}{\A} \in \XACU_2+$,
for all HA language $L$, $\pre^*_{\ptrs{\R}{\A}}(L)$
is a HA  language.
\end{theorem}

\subsection{Local Inconsistency of ACP} 
\label{sec:inconsistency}
Following e.g.~\cite{Bravo08},
an ACP for XML updates can be defined
by a pair $(\ptrs{\R_a}{\A}, \ptrs{\R_f}{\A})$ of PTRS, where
$\R_a$ contains allowed operations and $\R_f$ contains forbidden operations.
Such an ACP is called \emph{inconsistent}~\cite{FundulakiManeth07,Bravo08} 
if some forbidden operation can be simulated through a sequence of allowed operations.
\begin{example}
Assume that in the $\mathsf{hospital}$ document of example~\ref{ex:hospital}, 
it is forbidden to rename a $\mathsf{patient}$, that is the following update of
$\RPL_2$ is forbidden:
$\mathsf{patient}(y\, \mathsf{name}(x)\,z) \to \mathsf{patient}(y\, p_\mathsf{n}\, z)$.

\noindent If the following updates are allowed: 
$\mathsf{patient}(x) \to ()$ for deleting 
a $\mathsf{patient}$, and $\mathsf{hospital}(x) \to \mathsf{hospital}(x\,p_{\mathsf{pa}})$ 
to insert a $\mathsf{patient}$, then we have an inconsistency in the sense of ~\cite{Bravo08} 
since the effect of the forbidden update can be obtained by a combination of allowed updates. 
\end{example}
Using the results of Section~\ref{sec:update},
we can decide the above problems individually for terms of $D$.
More precisely, we solve the following problem called
\emph{local inconsistency}:
given a HA $\A$ over $\F$, 
an ACP $(\ptrs{\R_a}{\A}, \ptrs{\R_f}{\A})$ and a term $t \in \T(\F)$,
does there exists $u \in \T(\F)$ 
such that $t \lrstep{}{\ptrs{\R_f}{\A}} u$ and $t \lrstep{*}{\ptrs{\R_a}{\A}} u$?

\begin{theorem} \label{th:inconsistency}
Local inconsistency is decidable in PTIME for $\XACU+$.
\end{theorem}
\begin{proof}
It can be easily shown that the set $\{ u \in \T(\F) \mid t \lrstep{}{\ptrs{\R_f}{\A}} u \}$
is the language of a HA of size polynomial and constructed in PTIME on 
the  sizes of $\A$, $\R_f$ and $t$.
By Theorem~\ref{th:postXACU+}, $\post^*_{\ptrs{\R_a}{\A}}(\{ t \})$ is the language of a CF-HA
of polynomial size and constructed in polynomial time on the sizes of $\A$, $\R_a$ and $t$.
The ACP is locally inconsistent wrt $t$ iff the intersection of the two above language is
non empty, and this property can be tested in polynomial time.\qed
\end{proof}

\section*{Conclusion}
We have proposed a model for XML updates based on term rewriting,
and shown that type inference is possible
and the problems of reachability and typecheking are decidable
for the arbitrary application of $\XACU$ update rules, as well as some extensions,
when the application is only controlled by the label of the node at the update position 
and also at its parent node.
We have also shown that these problems become undecidable
when restricting the application of update operations 
to documents conforming to a fixed given DTD. 
\remarque{DTD or regular tree language ?}

As further works, we could study restrictions on the
regular tree languages in the constraints of PGTRS enabling 
the decidability of typechecking for $\XACU$ rules
with global constraints.
Another interesting topic, w.r.t. the verification 
ACP for updates based on annotated DTDs
is the access conditioned with XPath queries.
We could model this with rewrite rules constrained by XPath qualifiers.
Reachability is undecidable for such a formalism,
even when the rules are ground (a consequence of a result 
of~\cite{FundulakiManeth07}\footnote{Actually in~\cite{FundulakiManeth07}, the undecidability of the inconsistency problem
is stated but the construction in this paper proves the undecidability of reachability as well.}).
However, the construction of~\cite{FundulakiManeth07} involves upward navigation;
some fragments of downward Core XPath could permit to obtain decidability.



\newpage
\appendix
\renewcommand{\thesection}{\Alph{section}}

\section{Appendix: proof of Lemma~\ref{lem:collapsing}}

In this proof and the following, 
we describe the \emph{CF grammars} 
used for defining the horizontal languages of CF-HA transitions
as tuples
$\G = (\Sigma, \N, I, \Gamma)$,  
where $\Sigma$ is a finite alphabet (set of terminal symbols), 
$N$ is a set of non terminal symbols, 
$I \in \N$ is the initial non-terminal, 
and $\Gamma \in \N \times (\N \cup \Sigma)^*$ is a set of production rules.

\paragraph{\textsc{Lemma}~\ref{lem:collapsing}~\cite{JR-rta2008}.}
{\it For every extended CF-HA over $\F$ with collapsing transitions $\A$, 
there exists a CF-HA $\A'$ without collapsing transitions such that 
$L(\A') \cap \T(\F) = L(\A) \cap \T(\F)$.}

\medskip
\begin{RR}
\begin{proof}
Let $\G = (Q, N, I, \Gamma)$ and $\G_1 = (Q, N_1, I_1, \Gamma_1)$ 
be two CF grammars over the same finite alphabet $Q$.
Below, $\G$ and $\G_1$ are respectively meant to generate the languages $L$ and $L_1$ 
of CF HA transitions $L \to q$ and $a(L_1) \to p$.
We assume \textit{wlog} that the sets of non terminals $N$ and $N_1$ of 
$\G$ and $\G_1$ respectively are disjoint.
Let $q \in Q$ be a terminal symbol and let $X_q$ be a fresh non terminal symbol.
We consider below the CF grammar
\[ \G_1\mathclose{\downarrow}_q^{\G} :=
\bigl(Q, N_1 \uplus N \uplus \{ X_q \}, I_1, 
 \Gamma_1[q \leftarrow X_q] \cup \Gamma[q \leftarrow X_q] \cup 
 \{ X_q := q, X_q := I \} \bigr)
\]
where $\Gamma[q \leftarrow X_q]$ denotes the set of production rules of $\Gamma$
where every occurrence of the terminal symbol $q$ is replaced by the non-terminal $X_q$.
Using this construction, we can get rid of collapsing transitions in CF HA.

We assume that $\A$ is normalized with state set $Q$ 
and for each $a \in \F$ and $p \in Q$, we let $\G_{a,p}$
by the CF grammar generating the language $L_{a, p}$ in the transition
(assumed unique) $a(L_{a, p}) \to p$ of $\A$.
In order to construct $\A'$ out of $\A$,
we perform the following operation for every
collapsing transition $L \to q$ of $\A$:
(i.) delete $L \to q$
(ii.) for each $a \in \F$ and $p \in Q$, 
replace $\G_{a,p}$ by 
$\G_{a,p}\mathclose{\downarrow}_{q}^{\G}$
where $\G$ is a CF grammar generating $L$.
\qed
\end{proof}
\end{RR}

%
%

\section{Appendix: proof of Theorem~\protect\ref{th:post}}
\label{app:post*}
In this proof and the following, 
we describe \emph{finite automata} for the horizontal languages of HA transitions
as tuples
$B = (\Sigma, S, i, F, \Gamma)$, where 
$\Sigma$ is the finite input alphabet, 
$S$ is a finite set of states, $i$ is the initial state, 
$F \subseteq S$ is the set of final states and
$\Gamma \subseteq S \times (\Sigma \cup \{ \varepsilon \}) \times S$ is the set of transitions
and $\varepsilon$-transitions.
For $s, s' \in S$, we write $s \lrstep{\varepsilon}{B} s'$
to express that $s'$ can be reached from $s$ by a sequence of $\varepsilon$-transitions of $B$,
and $s \lrstep{a_1\ldots a_n}{} s'$,
for $a_1,\ldots, a_n \in \Sigma$,
if there exists $2(n+1)$ states $s_0, s'_0, \ldots, s_{n}, s'_n \in S$
with $s_0 = s$, $s_n  \lrstep{\varepsilon}{B} s'$ and
$0 \leq i < n$, 
$s_i \lrstep{\varepsilon}{B} s'_{i}$ and
$( s'_i, \sigma_{i+1}, s_{i+1}) \in \Gamma$.

\paragraph{\textsc{Theorem}~\ref{th:post}.}
{\it
Given a HA $\A$ on $\Sigma$ and a PTRS $\ptrs{\R}{\A} \in \XACU$,
for all HA language $L$, $\post_{\ptrs{\R}{\A}}^*(L)$ 
is the language of an HA of size polynomial 
and which can be constructed in PTIME on 
the size of $\R$, $\A$ and an HA of language $L$. }

\medskip
\begin{proof}
Let $\A = (\F, P, P^\final, \Theta)$ and
let  $\A_L = (\F, Q_L, Q_L^\final, \Delta_L)$ recognize $L$.
We assume that both $\A$ and $\A_L$ are normalized and
that their state sets $P$ and $Q_L$ are disjoint.
We construct a HA 
$\A' = (P \uplus Q_L, Q_L^\final, \Delta')$ 
recognizing $\post^*_{\ptrs{\R}{\A}}(L)$. 
%
For each $a \in \F$, $q \in Q_L$, let $L_{a, q}$ be the regular
language in the transition (assumed unique) $a(L_{a, q}) \to q \in \Delta_L$,
and let 
$B_{a, q} = \bigl(Q_L, S_{a, q}, \init_{a, q}, \{ f_{a, q} \}, \Gamma_{a, q}\bigr)$ 
be finite automaton recognizing $L_{a, q}$.
\begin{RR}
It has input alphabet $Q_L$, set of states $S_{a, q}$, 
initial state $\init_{a, q} \in S_{a, q}$, 
final state $f_{a, q} \in S_{a, q}$ (that we assume unique wlog)
and set of transition rules 
$\Gamma_{a, q} \subseteq S_{a, q} \times Q_L \times S_{a, q}$.
\end{RR}
The sets of states $S_{a, q}$ are assumed pairwise disjoint.
Let $S$ be the disjoint union of all $S_{a,q}$ for all $a \in \F$ and $q\in Q_L$.

For the construction of $\Delta'$, 
we develop a set of transition rules 
$\Gamma' \subseteq S \times (P \cup Q_L) \times S$.
Initially, we let $\Gamma'$ be the union $\Gamma_0$ of all $\Gamma_{a,q}$
for $a \in \Sigma$, $q\in Q_L$, and we complete $\Gamma'$ iteratively 
by analyzing the different cases of update rules of $\ptrs{\R}{\A}$.
At each step,
for each $a \in \F$ and $q \in Q_L$, 
we let $B'_{a, q}$ be the automaton $(P \cup Q_L, S, \init_{a,q}, \{ f_{a, q} \}, \Gamma')$.
For the sake of conciseness we make no distinction
between an automaton $B'_{a, q}$ and its language $L(B'_{a, q})$.


\begin{description}
\item{$\REN$:} for every $a(x) \to b(x) \in \ptrs{\R}{\A}$ and $q \in Q_L$, 
we add two $\varepsilon$-transitions
$(\init_{b,q}, \varepsilon, \init_{a,q})$ and
$(f_{a,q}, \varepsilon, f_{b,q})$ to $\Gamma'$.

\item{$\INS_\mathsf{first}$:}  for every $a(x) \to a(p\, x) \in \ptrs{\R}{\A}$ 
and $q \in Q_L$, 
we add one looping transition $(i_{a,q}, p, i_{a,q})$ to $\Gamma'$.

\item{$\INS_\mathsf{last}$:}  for every $a(x) \to a(x\, p) \in \ptrs{\R}{\A}$
and $q \in Q_L$, 
we add one looping transition rule
$(f_{a,q}, p, f_{a,q})$ to $\Gamma'$.

\item{$\INS_\mathsf{into}$:}  for every $a(xy) \to a(x\, p\, y) \in \ptrs{\R}{\A}$,
$q \in Q_L$ 
and $s \in S$ reachable from $i_{a,q}$ using the transitions of $\Gamma'$,
we add one looping transition rule $(s, p, s)$ to $\Gamma'$.

\item{$\INS_\mathsf{left}$:} for every $a(x) \to p\, a(x) \in \ptrs{\R}{\A}$, 
$q \in Q_L$ and state $s \in S$ such that 
$L(B'_{a,q}) \neq \emptyset$ and
there exists a transition $(s, q, s') \in \Gamma'$,
we add one looping transition $(s, p, s)$ to $\Gamma'$.

\item{$\INS_\mathsf{right}$:}  for every $a(x) \to a(x)\, p \in \ptrs{\R}{\A}$, 
$q \in Q_L$ and $s' \in S$
such that $L(B'_{a,q}) \neq \emptyset$
and there exists a transition $(s, q, s') \in \Gamma'$,
we add one looping transition $(s', p, s')$ to $\Gamma'$.

\item{$\RPL$:}  for every $a(x) \to p \in \ptrs{\R}{\A}$,
$q \in Q_L$, 
and $s, s' \in S$
such that $L(B'_{a,q}) \neq \emptyset$, 
and there exists a transition $(s, q, s') \in \Gamma'$,
we add one transition $(s, p, s')$ to $\Gamma'$.

\item{$\DEL$:}  for every $a(x) \to () \in \ptrs{\R}{\A}$,
$q \in Q_L$, and $s, s' \in S$
such that $L(B'_{a,q}) \neq \emptyset$, 
and there exists a transition $(s, q, s') \in \Gamma'$,
we add one $\varepsilon$-transition $(s, \varepsilon, s')$ to $\Gamma'$.
\end{description}

We iterate the above operations until a fixpoint is reached
(only a finite number of transition can be added to $\Gamma'$ this way).
Finally, we let 
\( \Delta' := \Theta \cup
\bigl\{ a\bigl(B'_{a, q}\bigr) \to q \bigm| a \in \F, q \in Q, 
              L(B'_{a, q}) \neq \emptyset \bigr\} \).
\noindent Let us show now
that $L(\A') = \post^*_{\ptrs{\R}{\A}}(L)$. 

\begin{RR}
\begin{lemma}
$L(\A') \subseteq \post^*_{\ptrs{\R}{\A}}(L)$.
\end{lemma}
\begin{proof}
We show more generally that for all $t \in L(\A', q)$, $q \in Q_L$, 
there exists $u \in L(\A_L, q)$ such that $u \lrstep{*}{\R} t$.
The proof is by induction on the multiset $\M$ of the applications
of horizontal transitions of $\Gamma'$ not in $\Gamma_0$ in a run 
of $\A'$ on $t$ leading to state $q$.

\paragraph{Base case.} If all the horizontal transitions are in $\Gamma_0$,
then by construction $t \in L(\A_L, q)$ and we are done.

\paragraph{Induction step.} We analyse the cases causing
the addition of a transition of $\Gamma' \setminus \Gamma_0$.

\paragraph{$\REN$}: 
let $t \in L(\A', q)$ ($q \in Q_L$),
and assume that an $\varepsilon$-transition 
$(\init_{b,q}, \varepsilon, \init_{a,q})$ is used in a run of $\A'$ on $t$,
and that this $\varepsilon$-transition was added to $\Gamma'$ 
because $a(x) \to b(x) \in \ptrs{\R}{\A}$.
Let 
\[ t = t[b(h)] \lrstep{*}{\A'} t[b(q_1\ldots q_n)] \lrstep{}{\A'} t[q_0] \lrstep{*}{\A'} q \]
be a reduction of $\A'$
such that the above $\varepsilon$-transition is involved  
in the step $t[b(q_1\ldots q_n)] \lrstep{}{\A'} t[q_0]$,
where the the transition $b(B_{b,q_0}) \to q_0$ is applied.
Hence $q_1\ldots q_n \in L(B_{b,q_0})$,
with $\init_{b,q} \lrstep{q_1\ldots q_n}{B_{b,q_0}} f_{b,q}$,
and the first step in this computation is $(\init_{b,q}, \varepsilon, \init_{a,q})$.
The last step must be $(f_{a,q}, \varepsilon, f_{b,q})$,
using an $\varepsilon$-transition
added to $\Gamma'$ in the same step as $(\init_{b,q}, \varepsilon, \init_{a,q})$.
By deleting these first and last steps, we get
$\init_{a,q} \lrstep{q_1\ldots q_n}{B_{a,q_0}} f_{a,q}$,
hence $q_1\ldots q_n \in L(B_{a,q_0})$.
Therefore, we have a reduction
$t'= t[a(h)] \lrstep{*}{\A'} t[a(q_1\ldots q_n)] \lrstep{}{\A'} t[q_0] \lrstep{*}{\A'} q$
(hence $t' \in L(\A', q)$) with a measure $\M$ 
strictly smaller than the above reduction for the recognition of $t$.
By induction hypothesis, it follows that
there exists $u \in L(\A_L, q)$ such that $u \lrstep{*}{\ptrs{\R}{\A}} t'$.
Since $t' = t[a(h)] \lrstep{}{\ptrs{\R}{\A}} t[b(h)] = t$, with $a(x) \to b(x)$,
we conclude that $u \lrstep{*}{\ptrs{\R}{\A}} t$.

\paragraph{$\INS_{\mathsf{first}}$}: 
let $t \in L(\A', q)$ ($q \in Q_L$),
and assume that an transition 
$(\init_{a,q}, p, \init_{a,q})$ is used in a run of $\A'$ on $t$,
and that this transition was added to $\Gamma'$ 
because $a(x) \to a(p\,x) \in \ptrs{\R}{\A}$.
Let 
\[ t = t[a(t_p h)] \lrstep{*}{\A'} t[a(p\, q_1\ldots q_n)] 
   \lrstep{}{\A'} t[q_0] \lrstep{*}{\A'} q
\]
be a reduction of $\A'$, with $t_p \in L(\A,p)$,
such that the above transition is involved  
in the step $t[a(p\,q_1\ldots q_n)] \lrstep{}{\A'} t[q_0]$,
where the the transition $b(B_{a,q_0}) \to q_0$ is applied.
Hence $p\,q_1\ldots q_n \in L(B_{a,q_0})$,
with $\init_{a,q} \lrstep{p\,q_1\ldots q_n}{B_{a,q_0}} f_{a,q}$,
and the first step in this computation is $(\init_{a,q}, p, \init_{a,q})$.
By deleting this first step, we get
$\init_{a,q} \lrstep{q_1\ldots q_n}{B_{a,q_0}} f_{a,q}$,
hence $q_1\ldots q_n \in L(B_{a,q_0})$.
Therefore, we have a reduction
$t'= t[a(h)] \lrstep{*}{\A'} t[a(q_1\ldots q_n)] \lrstep{}{\A'} t[q_0] \lrstep{*}{\A'} q$
(hence $t' \in L(\A', q)$) with a measure $\M$ 
strictly smaller than the above reduction for the recognition of $t$.
By induction hypothesis, it follows that
there exists $u \in L(\A_L, q)$ such that $u \lrstep{*}{\R} t'$.
Since $t' = t[a(h)] \lrstep{}{\ptrs{\R}{\A}} t[a(t_p h)] = t$, with $a(x) \to b(x)$,
we conclude $u \lrstep{*}{\R} t$.

\paragraph{$\INS_{\mathsf{last}}$}: this case is similar to the previous one.

\paragraph{$\INS_{\mathsf{into}}$}: 
let $t \in L(\A', q)$ ($q \in Q_L$),
and assume that an transition 
$(s, p, s)$ is used in a run of $\A'$ on $t$,
and that this transition was added to $\Gamma'$ 
because $a(xy) \to a(xpy) \in \ptrs{\R}{\A}$.
Let 
\[ t = t[a(h\, t_p\,\ell)] \lrstep{*}{\A'} 
   t[a(q_1\ldots q_n \, p\, q'_1\ldots q'_m)] \lrstep{}{\A'} t[q_0] \lrstep{*}{\A'} q \]
be a reduction of $\A'$, with $t_p \in L(\A,p)$,
such that the above transition $(s, p, s)$ is involved  
in the step $t[a(q_1\ldots q_n \, p\,q'_1\ldots q'_m)] \lrstep{}{\A'} t[q_0]$,
where the transition $b(B_{a,q_0}) \to q_0$ is applied.
More precisely, assume that $q_1\ldots q_n \, p\,q'_1\ldots q'_m \in L(B_{a,q_0})$,
because $\init_{a,q} \lrstep{q_1\ldots q_n}{B_{a,q_0}} s
      \lrstep{p}{B_{a,q_0}} s \lrstep{q'_1\ldots q'_m}{B_{a,q_0}} f_{a,q}$.
By deleting the middle step $(s, p, s)$, 
we get $\init_{a,q} \lrstep{q_1\ldots q_n\, q'_1\ldots q'_m}{B_{a,q_0}} f_{a,q}$,
hence $q_1\ldots q_n\, q'_1\ldots q'_m \in L(B_{a,q_0})$.
Therefore, we have a reduction
$t'= t[a(h\ell)] \lrstep{*}{\A'} t[a(q_1\ldots q_n\, q'_1\ldots q'_m)] 
 \lrstep{}{\A'} t[q_0] \lrstep{*}{\A'} q$
(hence $t' \in L(\A', q)$) with a measure $\M$ 
strictly smaller than the above reduction for the recognition of $t$.
By induction hypothesis, it follows that
there exists $u \in L(\A_L, q)$ such that $u \lrstep{*}{\ptrs{\R}{\A}} t'$.
Since $t' = t[a(h\ell)] \lrstep{}{\ptrs{\R}{\A}} t[a(h\, t_p\, \ell)] = t$, with $a(xy) \to b(xpy)$,
we conclude that $u \lrstep{*}{\ptrs{\R}{\A}} t$.

\paragraph{$\INS_{\mathsf{left}}$}: 
let $t \in L(\A', q)$ ($q \in Q_L$),
and assume that an transition 
$(s, p, s)$ is used in a run of $\A'$ on $t$,
and that this transition was added to $\Gamma'$ 
because $a(x) \to p\,a(x) \in \ptrs{\R}{\A}$
and because there exists $(s,q_0,s') \in \Gamma'$ 
for some $q_0 \in Q_L$
with $L(B_{a,q_0}) \neq \emptyset$.
Let 
\[ t = t[t_p\, a(h)] \lrstep{*}{\A'} t[p q_0] \lrstep{*}{\A'} q \]
be a reduction of $\A'$, with $t_p \in L(\A,p)$,
involving the transition $(s, p, s)$ in $s \lrstep{p q_0}{B_{b,q'}} s'$,
for some $b$.
Removing the transition $(s, p, s)$, we have $s \lrstep{q_0}{B_{b,q'}} s'$ 
and a reduction $t' = t[a(h)] \lrstep{*}{\A'} t[q_0] \lrstep{*}{\A'} q$
(meaning $t' \in L(\A', q)$) with a measure $\M$ 
strictly smaller than the above reduction for the recognition of $t$.
By induction hypothesis, it follows that
there exists $u \in L(\A_L, q)$ such that $u \lrstep{*}{\ptrs{\R}{\A}} t'$.
Since $t' = t[a(h)] \lrstep{}{\ptrs{\R}{\A}} t[t_p\, a(h)] = t$, with $a(x) \to p, a(x)$,
we conclude that $u \lrstep{*}{\ptrs{\R}{\A}} t$.

\paragraph{$\INS_{\mathsf{right}}$}: 
this case is similar to the previous one.

\paragraph{$\RPL$}: 
let $t \in L(\A', q)$ ($q \in Q_L$),
and assume that a horizontal transition 
$(s, p, s')$ is used in a run of $\A'$ on $t$,
and that this transition was added to $\Gamma'$ 
because $a(x) \to p \in \ptrs{\R}{\A}$
and because there exists $(s,q_0,s') \in \Gamma'$
for some $q_0 \in Q_L$ such that $L(B'_{a,q_0}) \neq \emptyset$.
Let \[ t = t[t_p] \lrstep{*}{\A'} t[p] \lrstep{*}{\A'} q \]
be a reduction of $\A'$, with $t_p \in L(\A,p)$,
involving the added transition $(s, p, s')$ in $s \lrstep{p}{B_{b,q'}} s'$,
for some $b$ and some $q' \in Q_L$.
Replacing the transition $(s, p, s')$ with $(s, q_0, s')$, 
we obtain $s \lrstep{q_0}{B_{b,q'}} s'$ 
and a reduction $t' = t[a(h)] \lrstep{*}{\A'} t[q_0] \lrstep{*}{\A'} q$
(meaning $t' \in L(\A', q)$).
The measure $\M$ of this later reduction is strictly smaller 
than the above reduction for the recognition of $t$, 
because the transition $(s, q_0, s')$
belongs to $\Gamma_0$ (no such transition can be added by the above procedure).
By induction hypothesis, it follows that
there exists $u \in L(\A_L, q)$ such that $u \lrstep{*}{\ptrs{\R}{\A}} t'$.
Since $t' = t[a(h)] \lrstep{}{\ptrs{\R}{\A}} t[t_p] = t$, with $a(x) \to p$,
we conclude that $u \lrstep{*}{\ptrs{\R}{\A}} t$.

\paragraph{$\DEL$}: 
let $t \in L(\A', q)$ ($q \in Q_L$),
and assume that a horizontal transition 
$(s, \varepsilon, s')$ is used in a run of $\A'$ on $t$,
and that this transition was added to $\Gamma'$ 
because $a(x) \to () \in \ptrs{\R}{\A}$
and because there exists $(s,q_0,s') \in \Gamma'$
for some $q_0 \in Q_L$ such that $L(B'_{a,q_0}) \neq \emptyset$.
Let us replace this $\varepsilon$-transition $(s, \varepsilon, s')$
with $(s,q_0,s')$ in a reduction $t \lrstep{*}{\A'} q$,
we obtain a reduction 
\[ t' = t[a(h)] \lrstep{*}{\A'} t[q_0] \lrstep{*}{\A'} q. \]
It means that $t' \in L(\A', q)$.
The measure $\M$ of this later reduction is strictly smaller 
than the above reduction for the recognition of $t$, 
because the transition $(s, q_0, s')$
belongs to $\Gamma_0$ (no such transition can be added by the above procedure).
By induction hypothesis, it follows that
there exists $u \in L(\A_L, q)$ such that $u \lrstep{*}{\ptrs{\R}{\A}} t'$.
Since $t' = t[a(h)] \lrstep{}{\ptrs{\R}{\A}} t$, with $a(x) \to ()$,
we conclude that $u \lrstep{*}{\ptrs{\R}{\A}} t$.

\hfill(end Lemma direction $\subseteq$)\qed
\end{proof}

\begin{lemma}
$L(\A') \supseteq \post^*_{\ptrs{\R}{\A}}(L)$.
\end{lemma}
\begin{proof}
We show that for all $t \in L$, 
if $t \lrstep{*}{\ptrs{\R}{\A}} u$, then $u \in L(\A')$,
by induction on the length of the rewrite sequence.

\paragraph{Base case $(0$ rewrite steps$)$.}
In this case, $u = t \in L$ and we are done since $L = L(\A_L) \subseteq L(\A')$ by construction.

\paragraph{Induction step.}
Assume that $t \lrstep{+}{\ptrs{\R}{\A}} u$ with $t \in L$. 
We analyse the type of rewrite rule used in the last rewrite step.

\paragraph{$\REN$.}
The last rewrite step of the sequence involves a rewrite rule of the form
$a(x) \to b(x) \in \ptrs{\R}{\A}$:
\[ 
 u \lrstep{*}{\ptrs{\R}{\A}} t[a(h)] \lrstep{}{\ptrs{\R}{\A}} t[b(h)] = t.
\]
By induction hypothesis, $t[a(h)] \in L(\A')$.
Hence there exists a reduction sequence:
\(
t[a(h)] \lrstep{*}{\A'} t[a(q_1 \ldots q_n)] \lrstep{}{\A'} t[q_0] 
        \lrstep{*}{\A'} q_\final \in Q^\final_L
\)
with $q_1 \ldots q_n \in L(B'_{a, q_0})$,
i.e. $\init_{a,q_0} \lrstep{q_1 \ldots q_n}{B'_{a, q_0}} f_{a,q_0}$.
By construction, the $\varepsilon$-transitions
$(\init_{b,q_0}, \varepsilon, \init_{a,q_0})$ and
$(f_{a,q_0}, \varepsilon, f_{b,q_0})$
have been added to $\Gamma'$.
Hence $\init_{b,q_0} \lrstep{q_1 \ldots q_n}{B'_{b, q_0}} f_{b,q_0}$
and $q_1 \ldots q_n \in L(B'_{b, q_0})$.
Therefore there exists a reduction sequence:
\(
t = t[b(h)] \lrstep{*}{\A'} t[b(q_1 \ldots q_n)] \lrstep{}{\A'} t[q_0] 
        \lrstep{*}{\A'} q_\final \in Q^\final_L
\)
and $t \in L(\A')$.

\paragraph{$\INS_\mathsf{first}$.}
The last rewrite step of the sequence involves a rewrite rule of the form
$a(x) \to a(p\,x) \in \ptrs{\R}{\A}$, with $p \in P$:
\[ 
 u \lrstep{*}{\ptrs{\R}{\A}} t[a(h)] \lrstep{}{\ptrs{\R}{\A}} t[a(t_p h)] = t
\]
with $t_p \in L(\A, p)$.
By induction hypothesis, $t[a(h)] \in L(\A')$.
Hence there exists a reduction sequence:
\(
t[a(h)] \lrstep{*}{\A'} t[a(q_1 \ldots q_n)] \lrstep{}{\A'} t[q_0] 
        \lrstep{*}{\A'} q_\final \in Q^\final_L
\)
with $q_1 \ldots q_n \in L(B'_{a, q_0})$,
i.e. $\init_{a,q_0} \lrstep{q_1 \ldots q_n}{B'_{a, q_0}} f_{a,q_0}$.
By construction, the transition
$(\init_{a,q_0}, p, \init_{a,q_0})$ has been added to $\Gamma'$.
Hence $\init_{a,q_0} \lrstep{p}{B'_{a, q_0}} \init_{a,q_0} 
       \lrstep{q_1 \ldots q_n}{B'_{a, q_0}} f_{b,q_0}$,
i.e. $p\,q_1 \ldots q_n \in L(B'_{a, q_0})$ and 
there exists a reduction sequence
\[
t = t[a(t_p\,h)] \lrstep{*}{\A'} t[a(p\,q_1 \ldots q_n)] \lrstep{}{\A'} t[q_0] 
        \lrstep{*}{\A'} q_\final \in Q^\final_L.
\]
It follows that $t \in L(\A')$.

\paragraph{$\INS_\mathsf{last}$.}
The case where the last rewrite step of the sequence involves a rewrite rule of the form
$a(x) \to a(x\, p) \in \ptrs{\R}{\A}$, with $p \in P$
is similar to the previous one.

\paragraph{$\INS_\mathsf{into}$.}
The last rewrite step of the sequence involves a rewrite rule of the form
$a(xy) \to a(x\,p\,y) \in \ptrs{\R}{\A}$, with $p \in P$:
\[ 
 u \lrstep{*}{\ptrs{\R}{\A}} t[a(h\ell)] \lrstep{}{\ptrs{\R}{\A}} t[a(h\,t_p\, \ell)] = t
\]
with $t_p \in L(\A, p)$.
By induction hypothesis, $t[a(h\ell)] \in L(\A')$.
Hence there exists a reduction sequence:
\(
t[a(h\ell)] \lrstep{*}{\A'} t[a(q_1 \ldots q_n\, q'_1 \ldots q'_m)] \lrstep{}{\A'} t[q_0] 
       \lrstep{*}{\A'} q_\final \in Q^\final_L
\)
with $q_1 \ldots q_n\, q'_1 \ldots q'_m \in L(B'_{a, q_0})$,
i.e. 
$\init_{a,q_0} \lrstep{q_1 \ldots q_n}{B'_{a, q_0}} s 
 \lrstep{q'_1 \ldots q'_m}{B'_{a, q_0}} f_{a,q_0}$ for some state $s \in S$.
By construction, the looping transition
$(s, p, s)$ has been added to $\Gamma'$.
Hence 
$\init_{a,q_0} \lrstep{q_1 \ldots q_n}{B'_{a, q_0}} s \lrstep{p}{B'_{a, q_0}} s
 \lrstep{q'_1 \ldots q'_m}{B'_{a, q_0}} f_{a,q_0}$, i.e.
 $q_1 \ldots q_n\, p\,q'_1 \ldots q'_m \in L(B'_{a, q_0})$ and 
there exists a reduction sequence
\[
t = t[a(h\,t_p\,\ell)] \lrstep{*}{\A'} t[a(q_1 \ldots q_n\, p\,q'_1 \ldots q'_m)] 
\lrstep{}{\A'} t[q_0] \lrstep{*}{\A'} q_\final \in Q^\final_L.
\]
It follows that $t \in L(\A')$.

\paragraph{$\INS_\mathsf{left}$.}
The last rewrite step of the sequence involves a rewrite rule of the form
$a(x) \to p\,a(x) \in \ptrs{\R}{\A}$, with $p \in P$:
\[ 
 u \lrstep{*}{\ptrs{\R}{\A}} t[a(h)] \lrstep{}{\ptrs{\R}{\A}} t[t_p\, a(h)] = t
\]
with $t_p \in L(\A, p)$.
By induction hypothesis, $t[a(h)] \in L(\A')$.
Hence there exists a reduction sequence:
\(
t[a(h)] \lrstep{*}{\A'} t[a(q_1 \ldots q_n)] \lrstep{}{\A'} t[q_0] 
       \lrstep{*}{\A'} q_\final \in Q^\final_L.
\)
Hence $L(B'_{a, q_0}) \neq \emptyset$
and at some point of the reduction, 
a transition $(s, q_0, s') \in \Gamma'$ is involved.
By construction, the transition
$(s, p, s)$ has been added to $\Gamma'$.
Hence 
there exists a reduction sequence
\(
t = t[t_p\, a(h)] \lrstep{*}{\A'} t[p\, q_0] \lrstep{*}{\A'} q_\final \in Q^\final_L.
\)
It follows that $t \in L(\A')$.

\paragraph{$\INS_\mathsf{right}$.}
The case where the last rewrite step of the sequence involves a rewrite rule of the form
$a(x) \to a(x)\, p \in \ptrs{\R}{\A}$, with $p \in P$
is similar to the previous one.

\paragraph{$\RPL$.}
The last rewrite step of the sequence involves a rewrite rule of the form
$a(x) \to p \in \ptrs{\R}{\A}$, with $p \in P$:
\[ 
 u \lrstep{*}{\ptrs{\R}{\A}} t[a(h)] \lrstep{}{\ptrs{\R}{\A}} t[t_p] = t
\]
with $t_p \in L(\A, p)$.
By induction hypothesis, $t[a(h)] \in L(\A')$.
Hence there exists a reduction sequence:
\(
t[a(h)] \lrstep{*}{\A'} t[a(q_1 \ldots q_n)] \lrstep{}{\A'} t[q_0] 
       \lrstep{*}{\A'} q_\final \in Q^\final_L.
\)
Hence $L(B'_{a, q_0}) \neq \emptyset$
and at some point of the reduction, 
a transition $(s, q_0, s') \in \Gamma'$ is applied.
By construction, the transition
$(s, p, s')$ has been added to $\Gamma'$,
and there exists a reduction sequence
\(
t = t[t_p] \lrstep{*}{\A'} t[p] \lrstep{*}{\A'} q_\final \in Q^\final_L.
\)
It follows that $t \in L(\A')$.

\paragraph{$\DEL$.}
The last rewrite step of the sequence involves a rewrite rule of the form
$a(x) \to () \in \ptrs{\R}{\A}$:
\[ 
 u \lrstep{*}{\ptrs{\R}{\A}} t[a(h)] \lrstep{}{\ptrs{\R}{\A}} t[()] = t.
\]
By induction hypothesis, $t[a(h)] \in L(\A')$.
Hence there exists a reduction sequence:
\(
t[a(h)] \lrstep{*}{\A'} t[a(q_1 \ldots q_n)] \lrstep{}{\A'} t[q_0] 
       \lrstep{*}{\A'} q_\final \in Q^\final_L.
\)
Hence $L(B'_{a, q_0}) \neq \emptyset$
and at some point of the reduction, 
a transition $(s, q_0, s') \in \Gamma'$ is applied.
By construction, the $\varepsilon$-transition
$(s, \varepsilon, s')$ has been added to $\Gamma'$,
and there exists a reduction sequence
\(
t \lrstep{*}{\A'} q_\final \in Q^\final_L,
\)
hence $t \in L(\A')$.

\hfill(end Lemma direction $\supseteq$) \qed
\end{proof}

\hfill (end of the proof of Theorem~\ref{th:post})
\end{RR}
\qed
\end{proof}

%
%

\section{Appendix: proof of Theorem~\ref{th:postXACU+}}
\label{app:post*XACU+}

\paragraph{\textsc{Theorem}~\ref{th:postXACU+}.}
{\it
Given a HA $\A$ on $\Sigma$ and a PTRS $\ptrs{\R}{\A} \in \XACU+$,
for all CF-HA term language $L$, $\post_{\ptrs{\R}{\A}}^*(L)$ 
is the language of an CF-HA of size polynomial 
and which can be constructed in PTIME on 
the size of $\R$, $\A$ and an CF-HA recognizing $L$. }

\medskip
\begin{proof}
Let $\A = (P, P^\final, \Theta)$ and let us assumed that it is normalized.
Let  $\A_L = (Q_L, Q_L^\final, \Delta_L)$ be a CF-HA recognizing $L$, 
normalized and without collapsing transitions
(this can be assumed thanks to Lemma~\ref{lem:collapsing})
The state sets $P$ and $Q_L$ are assumed disjoint.
We shall construct a CF-HA 
extended with collapsing transitions 
$\A' = (P \uplus Q_L, Q_L^\final, \Delta')$ 
recognizing $\post^*_{\ptrs{\R}{\A}}(L)$. 
The set of transitions $\Delta'$ is constructed starting 
from $\Delta_L \cup \Theta$
and analysing the different cases of update rules.

For each $a \in \F$, $q \in Q_L$, let $L_{a, q}$ be the context-free
language in the transition (assumed unique) $a(L_{a, q}) \to q \in \Delta_L$,
and let $\G_{a, q} = (Q_L, N_{a, q}, I_{a, q}, \Gamma_{a, q})$  be a CF grammar
in Chomski normal form generating $L_{a, q}$.
\begin{RR}
It has alphabet (set of terminal symbols) $Q_L$, 
set of non terminal symbols $N_{a, q}$, 
initial non-terminal $I_{a, q} \in N_{a, q}$, 
and set of production rules $\Gamma_{a, q}$.
\end{RR}
The sets of non-terminals $N_{a, q}$ 
are assumed pairwise disjoint.

Let us consider one new non-terminal $I'_{a, q}$ for each $a \in \F$ and $q \in Q_L$.
Each of these non terminals aims at becoming the initial non terminal 
of the CF grammar in the transition associated to $a$ and $q$ in $\Delta'$.
For technical convenience, we also add one new non terminal $X_p$
for each $p \in P$.
For the construction of $\Delta'$, we shall construct below 
a set $C'$ of collapsing transitions, initially empty, and
a set $\Gamma'$ of production rules of CF grammar over
the set of terminal symbols in $P \cup Q_L$ and the 
set of non terminals 
\[ \N = \displaystyle\bigcup_{a \in \F, q \in Q} \bigl( N_{a, q} \cup \{ I'_{a, q} \}\bigr)
 \cup \{ X_p \mid p \in P \}. \]
Initially, we let 
$\Gamma' = \Gamma'_0 := \displaystyle\bigcup_{a \in \F, q \in Q} 
  \bigl(P_{a, q} \cup \{ I'_{a,q} := I_{a,q} \}\bigr) 
  \cup \{ X_p := p \mid p \in P \}$.

We now proceed by analysis of the rewrite rules of $ \ptrs{\R}{\A}$
for the completion of $\Gamma'$ and $C'$.
At each step,
for each $a \in \F$ and $q \in Q_L$, 
we let $\G'_{a, q}$ be the CF grammar $(P \cup Q_L, \N, I'_{a, q}, \Gamma')$,
and let $L'_{a, q} = L(\G'_{a, q})$.
The production rules of $\Gamma'$ remain in Chomski normal form
after each completion step.

\begin{description}

\item{$\REN$:} for every $a(x) \to b(x) \in  \ptrs{\R}{\A}$,
$q \in Q_L$, 
we add one production rule $I'_{b,q} := I'_{a,q}$ to $\Gamma'$.

\item{$\INS'_\mathsf{first}$:} for every $a(x) \to b(p\, x) \in  \ptrs{\R}{\A}$, 
$q \in Q_L$, 
we add one production rule $I'_{b,q} := X_p I'_{a,q}$ to $\Gamma'$.

\item{$\INS'_\mathsf{last}$:} for every $a(x) \to b(x\, p) \in  \ptrs{\R}{\A}$, 
$q \in Q_L$, 
we add one production rule $I'_{b,q} := I'_{a,q} X_p$ to $\Gamma'$.

\item{$\INS_\mathsf{into}$:} for every $a(xy) \to a(x\, p\, y) \in  \ptrs{\R}{\A}$,
$q \in Q_L$
and every 
$N \in \N$ reachable from $I'_{a, q}$ using the rules of $\Gamma'$, 
we add two production rules $N := N X_p$ and $N := X_p N$. 


\item{$\INS_\mathsf{left}$:} for every $a(x) \to p\, a(x) \in  \ptrs{\R}{\A}$, 
and $q \in Q_L$ 
such that $L'_{a,q} \neq \emptyset$,
we add one collapsing transition $p\,q \to q$ to $C'$.

\item{$\INS_\mathsf{right}$:} for every $a(x) \to a(x)\, p \in \ptrs{\R}{\A}$, 
and $q \in Q_L$ 
such that $L'_{a,q} \neq \emptyset$, 
we add one collapsing transition $q\, p \to q$ to $C'$.

\item{$\RPL'$:} for every $a(x) \to p_1 \ldots p_n \in \ptrs{\R}{\A}$, with $n \geq 0$,
and $q \in Q_L$ such that $L'_{a,q} \neq \emptyset$, 
we add one collapsing transition $p_1 \ldots p_n \to q$ to $C'$.

\item{$\DEL$:} for every $a(x) \to () \in \ptrs{\R}{\A}$
and $q \in Q_L$ such that $L'_{a,q} \neq \emptyset$, 
we add one collapsing transition $() \to q$ to $C'$.

\end{description}
Note that $\INS_\mathsf{first}$, $\INS_\mathsf{last}$, $\RPL$
are special cases of respectively
$\INS'_\mathsf{first}$, $\INS'_\mathsf{last}$, $\RPL'$.

We iterate the above operations until a fixpoint is reached.
Indeed, only a finite number of production and collapsing rules
can be added. 
Finally, we let 
\[ \Delta' := \Theta \cup
\bigl\{ a(L'_{a, q}) \to q \bigm| a \in \F, q \in Q, L'_{a, q} \neq \emptyset \bigr\}
\cup C' \cup \{ L'_{a, q} \to q \mid a(x) \to x \in \ptrs{\R}{\A}, L'_{a,q} \neq \emptyset \}. \]

\noindent We show 
that $L(\A') = \post^*_{\ptrs{\R}{\A}}(L)$. 
It follows that $\post^*_{\ptrs{\R}{\A}}(L)$ is a CF-HA language
by Lemma~\ref{pr:collapse-CFHA}.

\remarque{sketch pr. correctness}
\begin{ABS}
The proof of the direction $\subseteq$ is by induction
on the number of application of collapsing transitions
in a reduction by $\A'$.
For the base case (no collapsing transition applied),
we make a second induction on the number of 
application of production rules of $\Gamma' \setminus \Gamma_0$
in the derivations, by the grammars $\G'_{a, q_0}$,
for the generations of the sequences of states $q_1\ldots q_n \in Q^*$
used in moves of $\A'$ of the form $C[a(q_1\ldots q_n)] \to C[q_0]$
in the reduction $t \lrstep{*}{\A'} q$.
Intuitively every application of such production rule
corresponds to a rewrite step with a rule of $ \ptrs{\R}{\A}$.
The proof of the direction $\supseteq$ is by induction 
on the length of a rewrite sequence $u \lrstep{*}{\ptrs{\R}{\A}} t$
for $u \in L(\A)$.
\end{ABS}
\begin{RR}
\begin{lemma}
$L(\A') \subseteq \post^*_{\ptrs{\R}{\A}}(L)$.
\end{lemma}
\begin{proof}
We show more generally that for all $t \in L(\A', q)$, $q \in Q_L$, 
there exists $u \in L(\A_L, q)$ such that $u \lrstep{*}{\ptrs{\R}{\A}} t$.
The proof is by induction on the number of 
applications of collapsing transitions 
in the reduction $t \lrstep{*}{\A'} q$.

\paragraph{Base case.}
For the base case (no collapsing transition applied),
we make a second induction on the number of 
application of production rules of $\Gamma' \setminus \Gamma_0$
in the derivations, by the grammars $\G'_{a, q_0}$,
for the generations of the sequences of states $q_1\ldots q_n \in Q^*$
used in moves of $\A'$ of the form $u[a(q_1\ldots q_n)] \to u[q_0]$
in the reduction $t \lrstep{*}{\A'} q$.
Let us note $\vdash$ the relation of derivation using 
the production rules of $\Gamma'$,
and $\vdash^*$ its transitive closure.

Intuitively every application of a production rule of 
$\Gamma' \setminus \Gamma_0$
corresponds to a rewrite step
with a rule of $ \ptrs{\R}{\A}$ in the rewrite sequence $u \lrstep{*}{\ptrs{\R}{\A}} t$,
according to the above construction cases.

\paragraph{Base case (second induction).}
\noindent For the base case, no production rule 
of $\Gamma' \setminus \Gamma_0$ is applied.
It means that $t \lrstep{*}{\A_L} q$ 
(every CF grammar derivation in the reduction $t \lrstep{*}{\A'} q$
 starts with $I'_{a, q} \vdash I_{a, q}$)
and we let $u = t$.

\paragraph{Induction step (second induction).}
Assume that the reduction $t \lrstep{*}{\A'} q$ has the form
\[
t = t[ a(t_1\ldots t_n) ] \lrstep{*}{\A'} t[a(q_1\ldots q_n)] \lrstep{}{\A'} 
    t[q_0] \lrstep{*}{\A'} q
\]
where $t[a(q_1\ldots q_n)] \lrstep{}{\A'} t[q_0]$ is one
transition such that the derivation of $I'_{a, q_0} \vdash^* q_1\ldots q_n$
by $\G'_{a, q_0}$ involves one 
production rule of $\Gamma' \setminus \Gamma_0$.
We shall analyse below the different cases of rewrite rules of $ \ptrs{\R}{\A}$
(rules of type $\XACU_1$)
which permitted the addition of this production rule of $\Gamma' \setminus \Gamma_0$.
Let us first note before that we can assume that for every $i \leq n$, 
$t_i \lrstep{*}{\A'} q_i$ because no collapsing transition are used, by hypothesis.
Hence, together with the above hypothesis, 
it follows that $t_i \in L(\A, q_i)$ for all $i \leq n$.

\paragraph{Case $\REN$.}
We have $I'_{a, q_0} \vdash I'_{b, q_0} \vdash^* q_1\ldots q_n$,
and the first production rule used in this derivation,
$I'_{a, q_0} := I'_{b, q_0}$,
was added because there exists a rule $b(x) \to a(x) \in  \ptrs{\R}{\A}$.
It follows that 
$I'_{b, q_0} \vdash^* q_1\ldots q_n$
and then that
\[ s = t[b(t_1\ldots t_n)] \lrstep{*}{\A'} 
   t\bigl[ b(q_1\ldots q_n) \bigr] \lrstep{}{\A'} t[q_0] \lrstep{*}{\A'} q.
\]
Hence, by induction hypothesis, 
there exists $u \in L(\A, q)$ 
such that $u \lrstep{*}{\ptrs{\R}{\A}} s$.
Moreover, 
$s = t[ b(t_1\ldots t_n) ] \lrstep{}{\ptrs{\R}{\A}} t = t[a(t_1\ldots t_n)]$
using $b(x) \to a(x) \in  \ptrs{\R}{\A}$.
Hence $u \lrstep{*}{\ptrs{\R}{\A}} t$.

\paragraph{Case $\INS'_\mathsf{first}$.}
We have 
$I'_{b, q_0} \vdash X_p I'_{a, q_0} \vdash^* q_1\ldots q_n$,
and the first production rule used in this derivation,
$I'_{b, q_0} := X_p I'_{a, q_0}$
was added because there exists a rule $a(x) \to b(p x) \in  \ptrs{\R}{\A}$.
%
By construction, it follows that $q_1 = p$ and 
$I'_{a, q_0} \vdash^* q_2\ldots q_n$,
and
\[ s = t[ a(t_2 \ldots t_n) ] \lrstep{*}{\A'} 
   t\bigl[ a(q_2 \ldots q_n) \bigr] \lrstep{}{\A'} t[q_0] \lrstep{*}{\A'} q.
\]
By induction hypothesis, 
applied to the above reduction, 
there exists $u \in L(\A, q)$ 
such that $u \lrstep{*}{\ptrs{\R}{\A}} s$.
Moreover, 
$s = t[ a(t_2\ldots t_n) ] \lrstep{}{\ptrs{\R}{\A}} t = t[b(t_1\ldots t_n)]$
using $a(x) \to a(p x) \in  \ptrs{\R}{\A}$, because $t_1 \in L(\A, p)$.
Hence $u \lrstep{*}{\ptrs{\R}{\A}} t$.

\paragraph{Case $\INS'_\mathsf{last}$.} This case is similar to the previous one.

\paragraph{Case $\INS_\mathsf{into}$.}
We have 
$I'_{a, q_0} \vdash^* \alpha N \beta \vdash \alpha N X_p\, \beta 
 \vdash \alpha N\, p \beta  \vdash^*  q_1\ldots q_n$,
and the production $N := N X_p$
was added because there exists a rule $a(xy) \to a(x p y) \in  \ptrs{\R}{\A}$,
and $N$  is reachable from $I'_{a,q}$ using $\Gamma'$.
%
%
It follows that there exists two integers $k < \ell \leq n$ 
such that 
$\alpha \vdash^* q_1\ldots q_k$ and
$N X_p \vdash^* q_{k+1} \ldots q_{\ell}$
(hence $q_{\ell} = p$)
and
$\beta \vdash^* q_{\ell+1} \ldots q_n$ 
(if $\ell = n$ then this latter sequence is empty),
and
\[ s = t[ a(t_1 \ldots t_{\ell-1}\, t_{\ell+1}\ldots t_n) ] \lrstep{*}{\A'} 
   t\bigl[ a(q_1\ldots q_{\ell-1}\, q_{\ell+1}\ldots q_n) \bigr] \lrstep{}{\A'} 
   t[q_0] \lrstep{*}{\A'} q.
\]
By induction hypothesis, 
applied to the above reduction, 
there exists $u \in L(\A, q)$ 
such that $u \lrstep{*}{\ptrs{\R}{\A}} s$.
Moreover, 
$s = t[ a(t_1 \ldots t_{\ell-1}\, t_{\ell+1}\ldots t_n) ] \lrstep{}{\ptrs{\R}{\A}} 
 t = t[a(t_1\ldots t_n)]$
using the rewrite rule $a(xy) \to a(x p y)$,
because $t_n \in L(\A, p)$.
Hence $u \lrstep{*}{\ptrs{\R}{\A}} t$.


\paragraph{Induction step (first induction).} 
Assume that the reduction $t \lrstep{*}{\A'} q$ has the form
\begin{equation}
 t = t[ t_1 \ldots t_n ] \lrstep{*}{\A'} t[ q_1\ldots q_n ] \lrstep{}{\A'} 
 t[q_0] \lrstep{*}{\A'} q
\label{eq:first-collapsing}
\end{equation}
such that there exists a collapsing transition
$L' \to q \in \Delta'$ with $q_1\ldots q_n \in L'$
and the first part of the reduction,
$t \lrstep{*}{\A'} t[ q_1\ldots q_n ]$,
involves no collapsing transition.
It implies in particular that $t_i \in L(\A', q_i)$ for all $i \leq n$.

The collapsing transition $L' \to q$ belongs to $C'$
(by hypothesis $\A_L$ and $\A$ do not contain collapsing transitions)
and was added because of a rewrite rule of $ \ptrs{\R}{\A}$ in $\XACU+$.
We consider below the different possible cases for this addition.

\paragraph{Case $\INS_\mathsf{left}$.}
We have $n = 2$, $q_1 = p \in P$, $q_2 = q_0$
and the collapsing transition $p q_0 \to q_0$ has been added because 
there exists a rule $a(x) \to p a(x) \in  \ptrs{\R}{\A}$. 
In this case, the reduction (\ref{eq:first-collapsing}) is
\[
 t = t[ t_1 t_2 ] \lrstep{*}{\A'} t[ p q_0 ] \lrstep{}{\A'} t[q_0] \lrstep{*}{\A'} q
\]
and we have
\(
 s = t[ t_2 ] \lrstep{*}{\A'} t[q_0] \lrstep{*}{\A'} q
\)
because the first part of the reduction uses no collapsing transition.
By induction hypothesis, there exists $u \in L(\A, q)$ 
such that $u \lrstep{*}{\ptrs{\R}{\A}} s$.
Moreover, 
$s \lrstep{}{\ptrs{\R}{\A}} t$
using the rewrite rule $a(x) \to p a(x)$,
because $t_1 \in L(\A, p)$.
Hence $u \lrstep{*}{\ptrs{\R}{\A}} t$.

\paragraph{Case $\INS_\mathsf{right}$.}
This case is similar to the previous one.

\paragraph{Case $\RPL'$.}
In this case, for all $i \leq n$, $q_i = p_i \in P$ and
the collapsing transition $p_1\ldots p_n \to q_0$ was added
because there exists a rewrite 
rule $a(x) \to p_1 \ldots p_n \in  \ptrs{\R}{\A}$
and $L'_{a, q_0} \neq \emptyset$.
Hence there exists a term $a(h) \in L(\A', q_0)$, and
\[
  s = t[ a(h) ] \lrstep{*}{\A'} t[q_0] \lrstep{*}{\A'} q
\]
By induction hypothesis, there exists $u \in L(\A, q)$ 
such that $u \lrstep{*}{\ptrs{\R}{\A}} s$.
Moreover, 
using the rewrite rule $a(x) \to p_1\ldots p_n$,
$s \lrstep{}{\ptrs{\R}{\A}} t$
because $t_i \in L(\A, p_i)$ for all $i \leq n$.
Hence $u \lrstep{*}{\ptrs{\R}{\A}} t$.

\paragraph{Case $\DEL$.}
In this case, $n=0$ and the collapsing transition 
$() \to q_0$ was added to $C'$
because there exists a rewrite rule $a(x) \to () \in  \ptrs{\R}{\A}$
and $L'_{a, q_0} \neq \emptyset$.
Let $a(h) \in L(\A', q_0)$, we have
\(
 s = t[ a(h) ] \lrstep{*}{\A'} t[q_0] \lrstep{*}{\A'} q.
\)
By induction hypothesis, there exists $u \in L(\A, q)$ 
such that $u \lrstep{*}{\ptrs{\R}{\A}} s$.
Moreover, 
$s \lrstep{}{\ptrs{\R}{\A}} t$
using the rewrite rule $a(x) \to ()$,
and $u \lrstep{*}{\ptrs{\R}{\A}} t$.

\paragraph{Case $\DEL_{\mathsf{s}}$.}
In this last case, the collapsing transition 
$L'_{a, q_0} \to q_0$ was added to $\Delta'$
because there exists a rewrite rule $a(x) \to x \in  \ptrs{\R}{\A}$
and $L'_{a, q_0} \neq \emptyset$.
We have 
\[
 s = t[ a(t_1\ldots t_n) ] \lrstep{*}{\A'} t[ a(q_1\ldots q_n) ] \lrstep{}{\A'} t[q_0] \lrstep{*}{\A'} q
\]
because $q_1\ldots q_n \in L'_{a, q_0}$.
By induction hypothesis, there exists $u \in L(\A, q)$ 
such that $u \lrstep{*}{\ptrs{\R}{\A}} s$.
Moreover, 
$s \lrstep{}{\ptrs{\R}{\A}} t$
using the rewrite rule $a(x) \to x$,
and $u \lrstep{*}{\ptrs{\R}{\A}} t$.

\hfill(end Lemma direction $\subseteq$)\qed
\end{proof}

\begin{lemma}
$L(\A') \supseteq \post^*_{\ptrs{\R}{\A}}(L)$.
\end{lemma}
\begin{proof}
We show that for all $u \in L$, 
if $u \lrstep{*}{\ptrs{\R}{\A}} t$, then $t \in L(\A')$,
by induction on the length of the rewrite sequence.

\paragraph{Base case $(0$ rewrite steps$)$.}
In this case, $u = t \in L$.
We can note that $L \subseteq L(\A')$
because $\Gamma'$ contains the production rule
$I'_{a, q} := I_{a, q}$ for all $a \in \Sigma$, $q \in Q_L$.
Hence, $t \in L(\A')$.

\paragraph{Induction step $(k+1$ rewrite steps$)$.}
We analyse the type of rewrite rule used in the last rewrite step
of $u \lrstep{*}{\ptrs{\R}{\A}} t$.

\paragraph{$\REN$.}
The last rewrite step of the sequence involves a rewrite rule of the form
$a(x) \to b(x) \in  \ptrs{\R}{\A}$:
\[ 
 u \lrstep{*}{\ptrs{\R}{\A}} u[a(h)] \lrstep{}{\ptrs{\R}{\A}} u[b(h)] = t.
\]
By induction hypothesis, $u[a(h)] \in L(\A')$.
Hence there exists a reduction sequence:
\(
u[a(h)] \lrstep{*}{\A'} u[a(q_1 \ldots q_n)] \lrstep{}{\A'} u[q_0] 
        \lrstep{*}{\A'} q_\final \in Q^\final_L
\)
with $q_1 \ldots q_n \in L'_{a, q_0}$,
i.e. $q_1 \ldots q_n$ can be generated by $\G'_{a,q_0}$,
starting from $I'_{a,q_0}$ and using the production rules of $\Gamma'$.

\noindent By construction, $\Gamma'$ contains 
the production rule $I'_{b, q_0} := I'_{a, q_0}$.
Hence $q_1 \ldots q_n \in L'_{b, q_0}$:
it can be generated by $\G'_{b,q_0}$,
starting from $I'_{b,q_0}$ and using the production rules of $\Gamma'$.

Hence 
\(
t = u[b(h)] \lrstep{*}{\A'} u[b(q_1 \ldots q_n)] \lrstep{}{\A'} u[q_0] 
        \lrstep{*}{\A'} q_\final \in Q^\final_L
\), i.e. $t \in L(\A')$.

\paragraph{$\INS'_\mathsf{first}$.}
The last rewrite step of the sequence involves a rewrite rule of the form
$a(x) \to b(p\,x) \in  \ptrs{\R}{\A}$, with $p \in P$:
\[ 
 u \lrstep{*}{\ptrs{\R}{\A}} u[a(h)] \lrstep{}{\ptrs{\R}{\A}} u[b(t_p h)] = t
\]
with $t_p \in L(\A, p)$.
By induction hypothesis, $u[a(h)] \in L(\A')$.
Hence there exists a reduction sequence:
\(
u[a(h)] \lrstep{*}{\A'} u[a(q_1 \ldots q_n)] \lrstep{}{\A'} u[q_0] 
        \lrstep{*}{\A'} q_\final \in Q^\final_L
\)
with $q_1 \ldots q_n \in L'_{a, q_0}$,
i.e. $q_1 \ldots q_n$ can be generated by $\G'_{a,q_0}$,
starting from $I'_{a,q_0}$ and using the production rules of $\Gamma'$.

\noindent By construction, $\Gamma'$ contains 
the production rule $I'_{b, q_0} := X_p I'_{a, q_0}$.
Hence $p q_1 \ldots q_n$ is in $L'_{b, q_0}$.
%
Hence 
\(
t = u[b(t_p h)] \lrstep{*}{\A'} u[b(p q_1 \ldots q_n)] \lrstep{}{\A'} u[q_0] 
        \lrstep{*}{\A'} q_\final \in Q^\final_L
\), i.e. $t \in L(\A')$.

\paragraph{$\INS'_\mathsf{last}$.}
This case is similar to the above one.

\paragraph{$\INS_\mathsf{into}$.}
The last rewrite step of the sequence involves a rewrite rule of the form
$a(xy) \to a(x p y) \in  \ptrs{\R}{\A}$, with $p \in P$:
\[ 
 u \lrstep{*}{\ptrs{\R}{\A}} u[a(h \ell)] \lrstep{}{\ptrs{\R}{\A}} u[a(h\, t_p\, \ell)] = t
\]
with $t_p \in L(\A, p)$.
By induction hypothesis, $u[a(h \ell)] \in L(\A')$.
Hence there exists a reduction sequence:
\(
u[a(h \ell)] \lrstep{*}{\A'} u[a(q_1 \ldots q_n)] \lrstep{}{\A'} u[q_0] 
        \lrstep{*}{\A'} q_\final \in Q^\final_L
\)
with $q_1 \ldots q_n \in L'_{a, q_0}$,
i.e. $q_1 \ldots q_n$ can be generated by $\G'_{a,q_0}$,
starting from $I'_{a,q_0}$ and using the production rules of $\Gamma'$.

\noindent By construction, $\Gamma'$ contains 
the production rules $N := N X_p$ and $N := X_p N$ for all
non terminal $N$ reachable from $I'_{q, q_0}$ using $\Gamma'$.
Using one of these production rules, 
it is possible to generate
$q_1 \ldots q_j\, p\, q_{j+1}\ldots q_n$ with $\G'_{a,q_0}$,
starting from $I'_{a,q_0}$ and using the production rules of $\Gamma'$,
where $j$ is the length of $h$.
Hence 
\(
t = u[a(h\, t_p\, \ell)] \lrstep{*}{\A'} u[b(q_1 \ldots q_j\, p\, q_{j+1}\ldots q_n)] 
    \lrstep{}{\A'} u[q_0] 
    \lrstep{*}{\A'} q_\final \in Q^\final_L
\), and $t \in L(\A')$.

\paragraph{$\INS_\mathsf{left}$.}
The last rewrite step of the sequence involves a rewrite rule of the form
$a(x) \to p a(x) \in  \ptrs{\R}{\A}$, with $p \in P$:
\[ 
 u \lrstep{*}{\ptrs{\R}{\A}} u[a(h)] \lrstep{}{\ptrs{\R}{\A}} u[t_p\, a(h)] = t.
\]
with $t_p \in L(\A,p)$.
By induction hypothesis, $u[a(h)] \in L(\A')$.
Hence there exists a reduction sequence:
\(
u[a(h)] \lrstep{*}{\A'} u[a(q_1 \ldots q_n)] \lrstep{}{\A'} u[q_0] 
        \lrstep{*}{\A'} q_\final \in Q^\final_L
\)
with $q_1 \ldots q_n \in L'_{a, q_0}$.

\noindent By construction, $\A'$ contains 
a collapsing transition rule $p q_0 \to q_0$.
Hence 
\(
t = u[t_p a(h)] \lrstep{*}{\A'} u[p q_0] \lrstep{}{\A'} u[q_0] 
        \lrstep{*}{\A'} q_\final \in Q^\final_L
\), \textit{i.e.} $t \in L(\A')$.

\paragraph{$\INS_\mathsf{right}$.}
This case is similar to the above one.

\paragraph{$\RPL'$.}
The last rewrite step of the sequence involves a rewrite rule of the form
$a(x) \to p_1\ldots p_n \in  \ptrs{\R}{\A}$, with $p_1,\ldots, p_n \in P$:
\[ 
 u \lrstep{*}{\ptrs{\R}{\A}} u[a(h)] \lrstep{}{\ptrs{\R}{\A}} u[t_1\ldots t_n] = t.
\]
with $t_i \in L(\A,p_i)$ for all $i \leq n$.
By induction hypothesis, $u[a(h)] \in L(\A')$.
Hence there exists a reduction sequence:
\(
u[a(h)] \lrstep{*}{\A'} u[a(q_1 \ldots q_n)] \lrstep{}{\A'} u[q_0] 
        \lrstep{*}{\A'} q_\final \in Q^\final_L
\)
with $q_1 \ldots q_n \in L'_{a, q_0}$.

\noindent Therefore, by construction, $\A'$ contains 
a collapsing transition rule $p_1\ldots p_n \to q_0$.
Hence 
\(
t = u[t_1\ldots t_n] \lrstep{*}{\A'} u[p_1 \ldots p_n] \lrstep{}{\A'} u[q_0] 
        \lrstep{*}{\A'} q_\final \in Q^\final_L
\), \textit{i.e.} $t \in L(\A')$.

\paragraph{$\DEL_{\mathsf{s}}$.}
The last rewrite step of the sequence involves a rewrite rule of the form
$a(x) \to () \in  \ptrs{\R}{\A}$:
\[ 
 u \lrstep{*}{\ptrs{\R}{\A}} u[a(h)] \lrstep{}{\ptrs{\R}{\A}} u[()] = t.
\]
By induction hypothesis, $u[a(h)] \in L(\A')$.
Hence there exists a reduction sequence:
\(
u[a(h)] \lrstep{*}{\A'} u[a(q_1 \ldots q_n)] \lrstep{}{\A'} u[q_0] 
        \lrstep{*}{\A'} q_\final \in Q^\final_L
\)
with $q_1 \ldots q_n \in L'_{a, q_0}$.

\noindent By construction, $\A'$ contains 
a collapsing transition rule $() \to q_0$.
Hence 
\(
t = u[()] \lrstep{}{\A'} u[q_0] \lrstep{*}{\A'} q_\final \in Q^\final_L
\), \textit{i.e.} $t \in L(\A')$.

\paragraph{$\DEL_{\mathsf{s}}$.}
The last rewrite step of the sequence involves a rewrite rule of the form
$a(x) \to x \in  \ptrs{\R}{\A}$:
\[ 
 u \lrstep{*}{\ptrs{\R}{\A}} u[a(h)] \lrstep{}{\ptrs{\R}{\A}} u[h] = t.
\]
By induction hypothesis, $u[a(h)] \in L(\A')$.
Hence there exists a reduction sequence:
\(
u[a(h)] \lrstep{*}{\A'} u[a(q_1 \ldots q_n)] \lrstep{}{\A'} u[q_0] 
        \lrstep{*}{\A'} q_\final \in Q^\final_L
\)
with $q_1 \ldots q_n \in L'_{a, q_0}$.

\noindent By construction, $\A'$ contains 
a collapsing transition rule $L'_{a, q_0} \to q_0$.
Hence 
\(
t = u[h] \lrstep{*}{\A'} u[q_1 \ldots q_n] \lrstep{}{\A'} u[q_0] 
        \lrstep{*}{\A'} q_\final \in Q^\final_L
\), \textit{i.e.} $t \in L(\A')$.

\hfill(end Lemma direction $\supseteq$)\qed
\end{proof}

\hfill (end of the proof of Theorem~\ref{th:post})
\end{RR}
\qed
\end{proof}

%
%

\section{Appendix: proof of Theorem~\ref{th:pre}}
\label{app:pre*}

\paragraph{\textsc{Theorem}~\ref{th:pre}.}
{\it
Given a HA $\A$ on $\Sigma$ and a PTRS $\ptrs{\R}{\A} \in \XACU+$,
for all HA language $L$, $\pre_{\ptrs{\R}{\A}}^*(L)$
is a HA the language.}

\medskip
\begin{proof}
Let $\A = (P, P^\final, \Theta)$,
and let  $\A_L = (Q_L, Q_L^\final, \Delta_L)$ be a HA recognizing $L$;
both are assumed normalized.
We also assume wlog that $\A_L$ is complete:
for all term $t$, there exists a state $q$ such that $t \in L(\A', q)$.
Like in the proof of Theorem~\ref{th:post},
we assume given, for each $a \in \F$, $q \in Q_L$,
a finite automaton 
$B_{a, q} = (Q_L, S_{a, q}, i_{a, q}, \{ f_{a, q} \}, \Gamma_{a, q})$
recognizing the regular language $L_{a, q}$ 
in the transition $a(L_{a, q}) \to q \in \Delta_L$ (assumed unique).

We shall construct a finite sequence sequence of HA 
$\A_0, \A_1, \ldots, \A_k$ whose final element's language is $\pre^*_{\ptrs{\R}{\A}}(L)$,
where for all $i \leq n$, 
$\A_i = (\F, Q_L, Q_L^\final, \Delta_i)$.
For the construction of the transition sets $\Delta_{i}$,
we consider a set $\C$ of finite automata over $Q_L$
defined as the smallest set such that:
\begin{itemize}
\item $\C$ contains every $B_{a, q}$ for $a \in \F$, $q \in Q_L$,
\item for all $B\in \C$, $B = (Q_L, S, i, F, \Gamma) \in \C$
and all states $s, s' \in S$, the automaton 
$B_{s,s'} := (Q_L, S, s, \{ s' \}, \Gamma)$ is in $\C$,
\item for all $B \in \C$, $B = (Q_L, S, i, F, \Gamma) \in \C$,
$q \in Q_L$ and all states $s, s' \in S$, the automata
$(Q_L, S, i, F, \Gamma \cup \{ \langle s,q,s'\rangle \})$ and
$(Q_L, S, i, F, \Gamma \cup \{ \langle s,\varepsilon,s'\rangle \})$,
respectively denoted by
$B+\langle s,q,s'\rangle$ and $B + \langle s,\varepsilon,s'\rangle $
also belong to $\C$.
\end{itemize}
Note that $\C$ is finite with this definition.
For the sake of conciseness, 
we make no distinction below between a NFA $B \in \C$
and the language $L(B)$ recognized by $B$.
Moreover, we assume that every $B \in \C$ has a unique final state
denoted $f_B$ and an initial state denoted $\init_B$.

First, we let $\Delta_0 = \Delta_L$.
%
The other $\Delta_i$
are constructed recursively by iteration
of the following case analysis until a fixpoint is reached
(only a finite number of transition can be added in the construction).
In the construction we use an extension of the move relation of HA,
from states to set of states (single states are considered as singleton sets):
$a(L_1,\ldots,L_n) \hookrightarrow_{\Delta_i} q$ 
(where $L_1,\ldots, L_n \subseteq Q_L$ and $q \in Q_L$)
iff there exists a transition $a(L) \to q \in \Delta_i$ 
such that $L_1\ldots L_n \subseteq L$.

\begin{description}
\item{$\REN$:} if $a(x) \to b(x) \in \ptrs{\R}{\A}$, 
 $B \in \C$ and $q \in Q_L$,
 such that $b(B) \hookrightarrow q$,
 then let $\Delta_{i+1} := \Delta_i \cup \{ a(B) \to q \}$.

\item{$\INS'_\mathsf{first}$:} 
 if $a(x) \to b(p\, x)  \in \ptrs{\R}{\A}$, $B\in \C$ and $q, q_p \in Q_L$,
 such that 
 $L(\A_i, q_p) \cap L(\A, p) \neq \emptyset$ and
 $b(q_p B) \hookrightarrow_{\Delta_i} q$,
 then $\Delta_{i+1} := \Delta_i \cup \{ a(B) \to q\}$.

\item{$\INS'_\mathsf{last}$:} 
 if $a(x) \to b(x\, p)  \in \ptrs{\R}{\A}$, $B\in \C$ and $q, q_p \in Q_L$,
 such that 
 $L(\A_i, q_p) \cap L(\A, p) \neq \emptyset$ and
 $b(B\, q_p) \hookrightarrow_{\Delta_i} q$,
 then $\Delta_{i+1} := \Delta_i \cup \{ a(B) \to q \}$.

\item{$\INS_\mathsf{into}$:}  
if $a(xy) \to a(x\, p\, y) \in \ptrs{\R}{\A}$,
$B\in \C$, $s, s'$ are states of $B$,
and $q, q_p \in Q_L$,
such that
$L(\A_i, q_p) \cap L(\A, p) \neq \emptyset$,
$s \lrstep{q_p}{B} s'$,
and $a(B) \hookrightarrow_{\Delta_i} q$
then 
$\Delta_{i+1} := \Delta_i \cup \bigl\{ a(B + \langle s, \varepsilon, s'\rangle) \to q \bigr\}$.

\item{$\INS_\mathsf{left}$:} 
if $a(x) \to p\, a(x) \in \ptrs{\R}{\A}$, 
$b \in \F$,
$B, B' \in \C$, $s, s'$ are states of $B$, and
$q, q_p ,q' \in Q_L$ 
such that 
$b(B) \to q \in \Delta_i$,
$a(B') \hookrightarrow_{\Delta_i} q'$,
$L(\A_i, q_p) \cap L(\A, p) \neq \emptyset$,
$s \lrstep{q_p q'}{B} s'$,
then 
$\Delta_{i+1} := \Delta_i \cup \bigl\{ b(B + \langle s, q', s'\rangle) \to q \bigr\}$.

\item{$\INS_\mathsf{right}$:} 
if $a(x) \to a(x)\, p \in \ptrs{\R}{\A}$, 
$b \in \F$,
$B, B' \in \C$, $s, s'$ are states of $B$, and
$q, q_p ,q' \in Q_L$ 
such that 
$b(B) \to q \in \Delta_i$,
$a(B') \hookrightarrow_{\Delta_i} q'$,
$L(\A_i, q_p) \cap L(\A, p) \neq \emptyset$,
$s \lrstep{q' q_p}{B} s'$,
then 
$\Delta_{i+1} := \Delta_i \cup\bigl\{ b(B + \langle s, q', s'\rangle) \to q \bigr\}$.

\item{$\RPL'$:} 
if $a(x) \to p_1\ldots p_n \in \ptrs{\R}{\A}$, 
$b \in \F$,
$B, B' \in \C$, $s, s'$ are states of $B$, and
$q, q', q_1,\ldots, q_n \in Q_L$ 
such that 
$b(B) \to q \in \Delta_i$,
$a(B') \hookrightarrow_{\Delta_i} q'$,
$L(\A_i, q_j) \cap L(\A, p_j) \neq \emptyset$ for all $1 \leq j \leq n$,
$s \lrstep{q_1\ldots q_n}{B} s'$
then 
$\Delta_{i+1} := \Delta_i \cup \bigl\{ b(B + \langle s, q', s'\rangle) ) \to q \bigr\}$.

\item{$\DEL$:} 
if $a(x) \to () \in \ptrs{\R}{\A}$,
$b \in \F$,
$B, B' \in \C$, $s$ is a state of $B$,
$q, q' \in Q_L$ 
such that 
$b(B) \to q \in \Delta_i$,
$a(B') \hookrightarrow_{\Delta_i} q'$,
then 
$\Delta_{i+1} := \Delta_i \cup \bigl\{ b(B + \langle s,q',s\rangle) \to q \bigr\}$.

\item{$\DEL_{\mathsf{s}}$:}
if $a(x) \to x \in \ptrs{\R}{\A}$,
$b \in \F$,
$B \in \C$, $s, s'$ are states of $B$,
$q, q' \in Q_L$ 
such that 
$b(B) \to q \in \Delta_i$,
$a(B_{s,s'}) \hookrightarrow_{\Delta_i} q'$,
then 
$\Delta_{i+1} := \Delta_i \cup \bigl\{ b(B + \langle s,q',s'\rangle) \to q \bigr\}$.

\end{description}

Note that 
$\INS_\mathsf{first}$, $\INS_\mathsf{last}$, $\RPL$
are special cases of respectively
$\INS'_\mathsf{first}$, $\INS'_\mathsf{last}$, $\RPL'$.
Since no state is added to the original automaton $\A_L$ and
all the transition added involve horizontal languages of the set $\C$,
which is finite, 
the iteration of the above operations terminates with an automaton $\A'$.
\noindent Let us show that $L(\A') = \pre^*_{\ptrs{\R}{\A}}(L)$. 

\begin{RR}
\begin{lemma}
$L(\A') \subseteq \pre^*_{\ptrs{\R}{\A}}(L)$.
\end{lemma}
\begin{proof}
We show more generally that for all $t \in L(\A', q)$, $q \in Q_L$, 
there exists $u \in L(\A_L, q)$ such that $t \lrstep{*}{\ptrs{\R}{\A}} u$.
The proof is by induction on the measure $\M$ associating to a reduction $t \lrstep{*}{\A'} q$
the multiset containing, for each transition rule 
$\rho \in \Delta_i$ with $i > 0$ used in the reduction,
the index $\mathit{min}( j > 0 \mid \rho \in \Delta_j)$. 

\paragraph{Base case.} If $\M$ is empty, all the transition are in $\Delta_0$.
It means that $t \in L(\A_L, q)$ and we let $u=t$.

\paragraph{Induction step.} 
Assume that we have a reduction by $\A'$ of the form
\begin{equation}
 t = t[a(h)] \lrstep{*}{\A'} t[a(q_1\ldots q_n)] \lrstep{}{\A'} t[q_0] \lrstep{*}{\A'} q 
 \label{eq:pre*-dir1}
\end{equation}
(with $q_0 \in Q_L$, $q_1\ldots q_n \in L(B)$)
and that the step $t[a(q_1\ldots q_n)] \lrstep{}{\A'} t[q_0]$ applies
a transition $b(B) \to q_0$ added to $\Delta_{i+1}$ for some $i \geq 0$.
We analyse the cases which permitted
the addition of this transition to $\Delta_{i+1}$.

\paragraph{$\REN$:} 
the transition $a(B) \to q_0$
was added to $\Delta_{i+1}$ because $a(x) \to b(x) \in \ptrs{\R}{\A}$
and $b(B) \hookrightarrow_{\Delta_i} q_0$.
Hence, there exists a reduction
\[ t' = t[b(h)] \lrstep{*}{\A'} t[b(q_1\ldots q_n)] \lrstep{}{\A'} t[q_0] \lrstep{*}{\A'} q \]
with a measure $\M$ strictly smaller than for (\ref{eq:pre*-dir1}), by hypothesis.
Therefore, by induction hypothesis, there exists $u \in L(\A_L, q)$ such that $t' \lrstep{*}{\ptrs{\R}{\A}} u$.
Since $t = t[a(h)] \lrstep{}{\ptrs{\R}{\A}} t[b(h)] = t'$, we conclude that $t \lrstep{*}{\ptrs{\R}{\A}} u$.

\paragraph{$\INS'_{\mathsf{first}}$:} 
the transition $a(B) \to q_0$ 
was added to $\Delta_{i+1}$ because $a(x) \to b(p\, x) \in \ptrs{\R}{\A}$,
with $q_0, q_p \in Q_L$, $L(\A_i, q_p) \cap L(\A, p) \neq \emptyset$ and $b(q_p B) \hookrightarrow_{\Delta_i} q_0$.
Hence, there exists a reduction
\[ t' = t[b(t_p\, h)] \lrstep{*}{\A'} t[b(q_p q_1\ldots q_n)] \lrstep{}{\A'} t[q_0] \lrstep{*}{\A'} q \]
with a measure $\M$ strictly smaller than for (\ref{eq:pre*-dir1}), by hypothesis.
Therefore, by induction hypothesis, there exists $u \in L(\A_L, q)$ such that $t' \lrstep{*}{\ptrs{\R}{\A}} u$.
Since $t = t[a(h)] \lrstep{}{\ptrs{\R}{\A}} t[b(t_p\, h)] = t'$, 
we conclude that $t \lrstep{*}{\ptrs{\R}{\A}} u$.

\paragraph{$\INS'_{\mathsf{last}}$:} this case is similar to the previous one.

\paragraph{$\INS_{\mathsf{into}}$:} 
the transition is $a(B') \to q_0$ 
and was added to $\Delta_{i+1}$ because $a(xy) \to b(x\, p\, y) \in \ptrs{\R}{\A}$,
$B \in \C$, 
$s, s'$ are states of $B$, $q_0, q_p \in Q_L$,
such that $L(\A_i, q_p) \cap L(\A, p) \neq \emptyset$,
$s \lrstep{q_p}{B} s'$,
$b(B) \hookrightarrow_{\Delta_i} q_0$
and $B' = B + \langle s, \varepsilon, s'\rangle$.
In this case, let $t = a(h \ell)$, and assume that the reduction (\ref{eq:pre*-dir1}) has the form
\[ t = t[a(h\ell)] \lrstep{*}{\A'} t[a(q_1\ldots q_n\, q'_1\ldots q'_m)] \lrstep{}{\A'} t[q_0] \lrstep{*}{\A'} q \]
with $q_1\ldots q_n\, q'_1\ldots q'_m \in L(B')$ by
$\init_{B'} \lrstep{q_1\ldots q_n}{B'} s \lrstep{\varepsilon}{B'} s' \lrstep{q'_1\ldots q'_m}{B'} f_{B'}$
($\init_{B'}$ and $f_{B'}$ are resp. initial and final states of $B'$).
Hence, by construction, we have
$\init_{B} \lrstep{q_1\ldots q_n}{B} s \lrstep{q_p}{B} s' \lrstep{q'_1\ldots q'_m}{B'} f_{B}$
($\init_{B'} = \init_{B}$ and $f_{B'} = f_{B}$)
and there exists a reduction
\[ t' = t[b(h\, t_p\, \ell)] \lrstep{*}{\A'} t[b(q_1\ldots q_n\, q_p\, q'_1\ldots q'_m)] 
       \lrstep{}{\A'} t[q_0] \lrstep{*}{\A'} q \]
with a measure $\M$ strictly smaller than for (\ref{eq:pre*-dir1}), by hypothesis.
Therefore, by induction hypothesis, there exists $u \in L(\A_L, q)$ such that $t' \lrstep{*}{\ptrs{\R}{\A}} u$.
Since $t = t[a(h\,\ell)] \lrstep{}{\ptrs{\R}{\A}} t[b(h\, t_p\, \ell)] = t'$, 
we conclude that $t \lrstep{*}{\ptrs{\R}{\A}} u$.

\medskip
\noindent From now on we assume that the reduction of $t$ by $\A'$ has the form
\begin{equation}
t = t[b(h)] \lrstep{*}{\A'} t[b(q_1\ldots q_n)] \lrstep{}{\A'} t[q_0] \lrstep{*}{\A'} q 
\label{eq:pre*-dir1b}
\end{equation}
with $q_1\ldots q_n \in L(B'')$, $q_0 \in Q_L$,
and that the step $t[b(q_1\ldots q_n)] \lrstep{}{\A'} t[q_0]$ applies
a transition $b(B'') \to q_0$ added to $\Delta_{i+1}$ for some $i \geq 0$
in one of the five cases.

\paragraph{$\INS_{\mathsf{left}}$:} 
the transition $b(B'') \to q_0$ 
was added to $\Delta_{i+1}$ because $a(x) \to p\,  a(x) \in \ptrs{\R}{\A}$,
$B, B' \in \C$, 
$s, s'$ are states of $B$, 
$q_0, q_p, q'_0 \in Q_L$,
such that
$b(B) \to q_0 \in \Delta_i$,
$a(B') \hookrightarrow_{\Delta_i} q'_0$,
$L(\A_i, q_p) \cap L(\A, p) \neq \emptyset$,
$s \lrstep{q_p\, q'_0}{B} s'$,
and $B'' = B + \langle s, q'_0, s' \rangle$.
In this case, let $t = b(h a(v) \ell)$, and assume that the above reduction (\ref{eq:pre*-dir1b}) has the form
\[ t = t[b(h\,a(v) \ell)] \lrstep{*}{\A'} t[b(q_1\ldots q_n\, q'_0\, q'_1\ldots q'_m)] \lrstep{}{\A'} t[q_0] \lrstep{*}{\A'} q \]
with $q_1\ldots q_n\, q'_1\ldots q'_m \in L(B'')$ by
$\init_{B''} \lrstep{q_1\ldots q_n}{B''} s \lrstep{q'_0}{B''} s' \lrstep{q'_1\ldots q'_m}{B''} f_{B''}$
($\init_{B''}$ and $f_{B''}$ are resp. the initial and final states of $B''$).
Hence, by construction, we have
$\init_{B} \lrstep{q_1\ldots q_n}{B} s \lrstep{q_p\, q'_0}{B} s' \lrstep{q'_1\ldots q'_m}{B} f_{B}$
($\init_{B''} = \init_{B}$ and $f_{B''} = f_{B}$)
and there exists a reduction
\[ t' = t[b(h\, t_p\, a(v)\, \ell)] \lrstep{*}{\A'} t[b(q_1\ldots q_n\, q_p\, q'_0  q'_1\ldots q'_m)] 
  \lrstep{}{\A'} t[q_0] \lrstep{*}{\A'} q \]
with a measure $\M$ strictly smaller than for (\ref{eq:pre*-dir1b}), by hypothesis.
Therefore, by induction hypothesis, there exists $u \in L(\A_L, q)$ such that $t' \lrstep{*}{\ptrs{\R}{\A}} u$.
Since $t = t[a(h\,a(v) \ell)] \lrstep{}{\ptrs{\R}{\A}} t[b(h\, t_p\,a(v) \ell)] = t'$, 
we conclude that $t \lrstep{*}{\ptrs{\R}{\A}} u$.

\paragraph{$\INS_{\mathsf{right}}$:} this case is similar to the previous one.

\paragraph{$\RPL'$:}
the transition $b(B'') \to q_0$ 
has been added to $\Delta_{i+1}$ because $a(x) \to p_1\ldots p_n \in \ptrs{\R}{\A}$,
$B, B' \in \C$, 
$s, s'$ are states of $B$, 
$q_0, q'_0, q_{p_1},\ldots, q_{p_n} \in Q_L$,
such that
$b(B) \to q_0 \in \Delta_i$,
$a(B') \hookrightarrow_{\Delta_i} q'_0$,
$L(\A_i, q_{p_j}) \cap L(\A, p_j) \neq \emptyset$ for all $j \leq n$,
$s \lrstep{q_{p_1}\ldots q_{p_n}}{B} s'$,
and $B'' = B + \langle s, q'_0, s' \rangle$.
In this case, let $t = b(h a(v) \ell)$, and assume that the above reduction (\ref{eq:pre*-dir1b}) has the form
\[ t = t[b(h\,a(v) \ell)] \lrstep{*}{\A'} t[b(q_1\ldots q_m\, q'_0\, q'_1\ldots q'_{m'})] 
      \lrstep{}{\A'} t[q_0] \lrstep{*}{\A'} q \]
with $q_1\ldots q_m\, q'_1\ldots q'_{m'} \in L(B'')$ by
$\init_{B''} \lrstep{q_1\ldots q_m}{B''} s \lrstep{q'_0}{B''} s' \lrstep{q'_1\ldots q'_{m'}}{B''} f_{B''}$
($\init_{B''}$ and $f_{B''}$ are resp. initial and final states of $B''$).
Hence, by construction, we have
$\init_{B} \lrstep{q_1\ldots q_m}{B} s \lrstep{q_{p_1}\ldots q_{p_n}}{B} s'
  \lrstep{q'_1\ldots q'_{m'}}{B} f_{B}$
($\init_{B''} = \init_{B}$ and $f_{B''} = f_{B}$)
and there exists a reduction,
with for all $j \leq n$, $t_j \in L(\A_i, q_{p_j}) \cap L(\A, p_j)$,
\[ t' = t[b(h\, t_1\ldots t_n\, \ell)] \lrstep{*}{\A'} 
   t[b(q_1\ldots q_m\, q_{p_1}\ldots q_{p_n}\,  q'_1\ldots q'_{m'})] \lrstep{}{\A'} 
   t[q_0] \lrstep{*}{\A'} q \]
with a measure $\M$ strictly smaller than for (\ref{eq:pre*-dir1b}), by hypothesis.
Therefore, by induction hypothesis, there exists $u \in L(\A_L, q)$ such that $t' \lrstep{*}{\ptrs{\R}{\A}} u$.
Since $t = t[a(h\,a(v) \ell)] \lrstep{}{\ptrs{\R}{\A}} t[b(h\, t_1\ldots t_n\, \ell)] = t'$, 
using the rule $a(x) \to p_1\ldots p_n$,
and we conclude that $t \lrstep{*}{\ptrs{\R}{\A}} u$.

\paragraph{$\DEL$:}
the transition $b(B'') \to q_0$ 
has been added to $\Delta_{i+1}$ because $a(x) \to () \in \ptrs{\R}{\A}$,
$B, B' \in \C$, 
$s$ is a state of $B$, 
$q_0, q'_0 \in Q_L$,
such that
$b(B) \to q_0 \in \Delta_i$,
$a(B') \hookrightarrow_{\Delta_i} q'_0$,
and $B'' = B + \langle s, q'_0, s \rangle$.
In this case, let $t = b(h\, a(v) \ell)$, and assume that the above reduction (\ref{eq:pre*-dir1b}) has the form
\[ t = t[b(h\,a(v) \ell)] \lrstep{*}{\A'} t[b(q_1\ldots q_m\, q'_0\, q'_1\ldots q'_{m'})] 
      \lrstep{}{\A'} t[q_0] \lrstep{*}{\A'} q \]
with $q_1\ldots q_m\, q'_1\ldots q'_{m'} \in L(B'')$ by
$\init_{B''} \lrstep{q_1\ldots q_m}{B''} s \lrstep{q'_0}{B''} s \lrstep{q'_1\ldots q'_{m'}}{B''} f_{B''}$
($\init_{B''}$ and $f_{B''}$ are resp. initial and final states of $B''$).
Hence, by construction, we have
$\init_{B} \lrstep{q_1\ldots q_m}{B} s \lrstep{q'_1\ldots q'_{m'}}{B} f_{B}$
($\init_{B''} = \init_{B}$ and $f_{B''} = f_{B}$)
and there exists a reduction
\[ t' = t[b(h\, \ell)] \lrstep{*}{\A'} 
   t[b(q_1\ldots q_m\, q'_1\ldots q'_{m'})] \lrstep{}{\A'} 
   t[q_0] \lrstep{*}{\A'} q \]
with a measure $\M$ strictly smaller than for (\ref{eq:pre*-dir1b}), by hypothesis.
Therefore, by induction hypothesis, there exists $u \in L(\A_L, q)$ such that $t' \lrstep{*}{\ptrs{\R}{\A}} u$.
Since $t = t[a(h\,a(v) \ell)] \lrstep{}{\ptrs{\R}{\A}} t[b(h\, \ell)] = t'$, 
and we conclude that $t \lrstep{*}{\ptrs{\R}{\A}} u$.

\paragraph{$\DEL_{\mathsf{s}}$:} 
the transition $b(B'') \to q_0$ 
has been added to $\Delta_{i+1}$ because 
$a(x) \to x \in \ptrs{\R}{\A}$,
$B\in \C$, 
$s, s'$ are states of $B$, 
$q_0, q'_0 \in Q_L$,
such that
$b(B) \to q_0 \in \Delta_i$,
$a(B_{s,s'}) \hookrightarrow_{\Delta_i} q'_0$,
and $B'' = B + \langle s, q'_0, s' \rangle$.
In this case, let $t = b(h a(v) \ell)$, and assume that the above reduction (\ref{eq:pre*-dir1b}) has the form
\[ t = t[b(h\,a(v) \ell)] \lrstep{*}{\A'} 
      t[b(q_1\ldots q_m\, a(v_1\ldots v_k)\, q'_1\ldots q'_{m'})] 
      \lrstep{*}{\A'} 
      t[b(q_1\ldots q_m\, q'_0\, q'_1\ldots q'_{m'})] 
      \lrstep{}{\A'} t[q_0] \lrstep{*}{\A'} q \]
with $q_1\ldots q_m\, q'_0, q'_1\ldots q'_{m'} \in L(B'')$ by
$\init_{B''} \lrstep{q_1\ldots q_m}{B''} s \lrstep{q'_0}{B''} s' \lrstep{q'_1\ldots q'_{m'}}{B''} f_{B''}$
($\init_{B''}$ and $f_{B''}$ are resp. initial and final states of $B''$)
and $s \lrstep{v_1\ldots v_k}{B_{s,s'}} s'$.

Hence, by construction, we have
$\init_{B} \lrstep{q_1\ldots q_m}{B} s \lrstep{v_1\ldots v_k}{B} s' 
           \lrstep{q'_1\ldots q'_{m'}}{B} f_{B}$
($\init_{B''} = \init_{B}$ and $f_{B''} = f_{B}$)
and there exists a reduction
\[ t' = t[b(h\, v\, \ell)] \lrstep{*}{\A'} 
   t[b(q_1\ldots q_m\, v_1\ldots v_k\, q'_1\ldots q'_{m'})] \lrstep{}{\A'} 
   t[q_0] \lrstep{*}{\A'} q \]
with a measure $\M$ strictly smaller than for (\ref{eq:pre*-dir1b}), by hypothesis.
Therefore, by induction hypothesis, there exists $u \in L(\A_L, q)$ such that $t' \lrstep{*}{\ptrs{\R}{\A}} u$.
Since $t = t[a(h\,a(v) \ell)] \lrstep{}{\ptrs{\R}{\A}} t[b(h\, v\, \ell)] = t'$, 
we conclude that $t \lrstep{*}{\ptrs{\R}{\A}} u$.

Note that $\INS_\mathsf{first}$, $\INS_\mathsf{last}$, $\RPL$,
were not considered above because
they are special cases of respectively
$\INS'_\mathsf{first}$, $\INS'_\mathsf{last}$, $\RPL'$.

\hfill(end Lemma direction $\subseteq$)\qed
\end{proof}

\begin{lemma}
$L(\A') \supseteq \pre^*_{\ptrs{\R}{\A}}(L)$.
\end{lemma}
\begin{proof}
We show that for all $t \in L$, 
if $u \lrstep{*}{\ptrs{\R}{\A}} t$, then $u \in L(\A')$,
by induction on the length of the rewrite sequence.

\paragraph{Base case $(0$ rewrite steps$)$.}
In this case, $u = t \in L$ and we are done since $L = L(\A_L) \subseteq L(\A')$ by construction.

\paragraph{Induction step.}
Assume that $u \lrstep{+}{\ptrs{\R}{\A}} t$, we analyse the type of rewrite rule used in the first rewrite step.

\paragraph{$\REN$.}
Assume that $u = u[a(h)] \lrstep{}{\ptrs{\R}{\A}} u[b(h)] \lrstep{*}{\ptrs{\R}{\A}} t$.
By induction hypothesis, $u_1 =  u[b(h)] \in L(\A')$,
i.e. there exists a reduction sequence 
$u_1 = u[b(h)] \lrstep{*}{\A'} u[b(q_1\ldots q_n)]  \lrstep{}{\A'} u[q]  \lrstep{*}{\A'} q^\final$
where $q, q_1,\ldots, q_n \in Q_L$, $q^\final \in Q^\final_L$,
and a transition $a(B) \to q$ has been added to $\A'$, with $q_1\ldots q_n \in B$.
It follows that 
$u = u[a(h)] \lrstep{*}{\A'} u[a(q_1\ldots q_n)] \lrstep{}{\A'} u[q] \lrstep{*}{\A'} q^\final$,
hence that $u \in L(\A')$.

\paragraph{$\INS'_{\mathsf{first}}$.}
Assume that $u = u[a(h)] \lrstep{}{\ptrs{\R}{\A}} u[b(t_p\, h)] \lrstep{*}{\ptrs{\R}{\A}} t$
for some $t_p \in L(\A, p)$.
By induction hypothesis, $u_1 =  u[b(t_p\, h)] \in L(\A')$,
i.e. there exists a reduction sequence 
\[ u[b(t_p\, h)] \lrstep{*}{\A'} u[b(q_p q_1\ldots q_n)]  \lrstep{}{\A'} u[q]  \lrstep{*}{\A'} q^\final \]
where $q, q_p, q_1,\ldots, q_n \in Q_L$, $q^\final \in Q^\final_L$.
Hence $L(\A', q_p) \cap L(\A, p)$ is not empty because it contains $t_p$,
and a transition $a(B) \to q$ has been added to $\A'$, with $q_1\ldots q_n \in B$.
It follows that 
$u = u[a(h)] \lrstep{*}{\A'} u[a(q_1\ldots q_n)] \lrstep{}{\A'} u[q] \lrstep{*}{\A'} q^\final$,
hence that $u \in L(\A')$.

\paragraph{$\INS'_{\mathsf{last}}$.} This case is similar to the previous one.

\paragraph{$\INS_{\mathsf{into}}$.}
Assume that $u = u[a(h\ell)] \lrstep{}{\ptrs{\R}{\A}} u[a(h\, t_p\, \ell)] \lrstep{*}{\ptrs{\R}{\A}} t$
for some $t_p \in L(\A, p)$.
By induction hypothesis, $u_1 =  u[a(h\, t_p\, \ell)] \in L(\A')$,
i.e. there exists a reduction sequence 
\[ u_1 = u[a(h\, t_p\, \ell)] \lrstep{*}{\A'} u[a(q_1\ldots q_m q_p q'_1\ldots q'_n)] 
\lrstep{\rho}{\A'} u[q]  \lrstep{*}{\A'} q^\final \]
where $q, q_p, q_1,\ldots, q_m, q'_1,\ldots, q'_n \in Q_L$ and $q^\final \in Q^\final_L$.
Hence $L(\A', q_p) \cap L(\A, p)$ is not empty because it contains $t_p$,
and the transition rule denoted $\rho$ in the above sequence 
has the form $b(B) \to q$,
where $q_1\ldots q_m q_p q'_1\ldots q'_n$ is recognized by $B$,
with a sequence 
$\init_B \lrstep{q_1\ldots q_m}{B} s \lrstep{q_p}{B} s' \lrstep{q'_1\ldots q'_n}{B} f_B$
for some states $s, s'$ of $B$.
Therefore, a transition $a(B + \langle s,\varepsilon,s'\rangle) \to q$ has been added to $\A'$, 
and $q_1\ldots q_m q'_1\ldots q'_n$ 
is recognized by $B + \langle s,\varepsilon,s'\rangle$.
It follows that 
$u = u[a(h\ell)] \lrstep{*}{\A'} u[a(q_1\ldots q_m q'_1\ldots q'_n)] \lrstep{}{\A'} u[q] \lrstep{*}{\A'} q^\final$,
hence that $u \in L(\A')$.

\paragraph{$\INS_{\mathsf{left}}$.}
Assume that $u = u[b(h\, a(v) \ell)] \lrstep{}{\ptrs{\R}{\A}} u[b(h\, t_p\, a(v) \ell)] \lrstep{*}{\ptrs{\R}{\A}} t$
for some $t_p \in L(\A, p)$.
By induction hypothesis, $u_1 =  u[b(h\, t_p\, a(v) \ell)] \in L(\A')$,
i.e. there exists a reduction sequence 
\[ u[b(h\, t_p\, a(v) \ell)] \lrstep{*}{\A'} u[b(q_1\ldots q_m\,  q_p\, q' q'_1\ldots q'_n)] 
\lrstep{\rho}{\A'} u[q]  \lrstep{*}{\A'} q^\final \]
where $q, q', q_p, q_1,\ldots, q_m, q'_1,\ldots, q'_n \in Q_L$, $q^\final \in Q^\final_L$, 
and $a(v) \lrstep{*}{\A'} q'$.
Hence $L(\A', q_p) \cap L(\A, p)$ is not empty because it contains $t_p$,
and the transition rule denoted $\rho$ in the above sequence 
has the form $b(B) \to q$ with
$q_1\ldots q_m  q_p q' q'_1\ldots q'_n$ is recognized by $B$,
with a sequence,
$\init_B \lrstep{q_1\ldots q_m}{B} s \lrstep{q_p q'}{} s' \lrstep{q'_1\ldots q'_n}{B} f_B$ 
for some of states $s$ and $s'$ of $B$.
Hence, a transition $b(B + \langle s,q',s'\rangle) \to q$ 
has been added to $\A'$, and $q_1\ldots q_m q' q'_1\ldots q'_n$ is recognized by $B + \langle s,q',s'\rangle$.
It follows that 
$u = u[b(h a(v)\ell)] \lrstep{*}{\A'} u[a(q_1\ldots q_m q' q'_1\ldots q'_n)] \lrstep{}{\A'} u[q] \lrstep{*}{\A'} q^\final$,
hence that $u \in L(\A')$.

\paragraph{$\INS_{\mathsf{right}}$.} This case is similar to the previous one.

\paragraph{$\RPL'$.}
Assume that $u = u[b(h a(v) \ell)] \lrstep{}{\ptrs{\R}{\A}} u[b(h t_1\ldots t_n \ell)] \lrstep{*}{\ptrs{\R}{\A}} t$
for some $t_1,\ldots, t_n$ respectively in $L(\A, p_1),\ldots, L(\A, p_n)$.
By induction hypothesis, $u_1 =  u[b(h t_1\ldots t_n \ell)] \in L(\A')$,
i.e. there exists a reduction sequence 
\[ u[b(h\, t_1\ldots t_n\, \ell)] \lrstep{*}{\A'} u[b(q_1\ldots q_m\,  q_{p_1}\ldots q_{p_n}\, q'_1\ldots q'_{m'})] 
\lrstep{\rho}{\A'} u[q]  \lrstep{*}{\A'} q^\final \]
where $q, q_{p_1},\ldots, q_{p_n}, q_1,\ldots, q_m, q'_1,\ldots, q'_{m'} \in Q_L$, $q^\final \in Q^\final_L$,
and for all $j \leq n$, $L(\A', q_{p_j}) \cap L(\A, p_j)$ contains $t_j$,
and the transition rule denoted $\rho$ in the above sequence has the form $b(B) \to q$ with
$q_1\ldots q_m\,  q_{p_1}\ldots q_{p_n}\, q'_1\ldots q'_{m'} \in L(B)$,
with a sequence 
$\init_B \lrstep{q_1\ldots q_m}{B} s \lrstep{q_{p_1}\ldots q_{p_n}}{} s' \lrstep{q'_1\ldots q'_{m'}}{B} f_B$, 
for some states $s$ and $s'$ of $B$.
Let $q' \in Q_L$ be such that $a(v) \lrstep{*}{\A'} q'$.
By construction, a transition $b(B + \langle s,q',s'\rangle) \to q$ 
has been added to $\A'$, and $q_1\ldots q_m\, q'\, q'_1\ldots q'_{m'}$ is recognized by $B + \langle s,q',s'\rangle$.
It follows that 
$u = u[b(h a(v)\ell)] \lrstep{*}{\A'} u[a(q_1\ldots q_m q' q'_1\ldots q'_{m'})] \lrstep{}{\A'} u[q] \lrstep{*}{\A'} q^\final$,
hence that $u \in L(\A')$.

\paragraph{$\DEL$.}
Assume that $u = u[b(h a(v) \ell)] \lrstep{}{\ptrs{\R}{\A}} u[b(h \ell)] \lrstep{*}{\ptrs{\R}{\A}} t$.
By induction hypothesis, $u_1 =  u[b(h \ell)] \in L(\A')$,
i.e. there exists a reduction sequence 
\[ u[b(h\ell)] \lrstep{*}{\A'} u[b(q_1\ldots q_m\, q'_1\ldots q'_{m'})] 
\lrstep{\rho}{\A'} u[q] \lrstep{*}{\A'} q^\final \]
where $q, q_1,\ldots, q_m, q'_1,\ldots, q'_{m'} \in Q_L$ and $q^\final \in Q^\final_L$.
The transition rule denoted $\rho$ in the above sequence has the form $b(B) \to q$ and
$q_1\ldots q_m\, q'_1\ldots q'_{m'}$ is recognized by $B$ with a sequence
$\init_B \lrstep{q_1\ldots q_m}{B} s \lrstep{q'_1\ldots q'_{m'}}{B} f_B$,
where $s$ is a state of $B$.
Let $q' \in Q_L$ be such that $a(v) \lrstep{*}{\A'} q'$.
By construction, a transition $b(B + \langle s, q',s\rangle) \to q$ has been added to $\A'$, 
and $q_1\ldots q_m\, q'\, q'_1\ldots q'_{m'}$ is recognized by $B + \langle s, q',s\rangle$.
It follows that 
$u = u[b(h\, a(v)\ell)] \lrstep{*}{\A'} u[a(q_1\ldots q_m q' q'_1\ldots q'_{m'})] \lrstep{}{\A'} u[q] \lrstep{*}{\A'} q^\final$,
hence that $u \in L(\A')$.

\paragraph{$\DEL_{\mathsf{s}}$.}
Assume that $u = u[b(h a(v) \ell)] \lrstep{}{\ptrs{\R}{\A}} u[b(h v \ell)] \lrstep{*}{\ptrs{\R}{\A}} t$.
By induction hypothesis, $u_1 =  u[b(h v \ell)] \in L(\A')$,
i.e. there exists a reduction sequence 
\[ u[b(h v \ell)] \lrstep{*}{\A'} u[b(q_1\ldots q_m\, q''_1\ldots q''_n \,q'_1\ldots q'_{m'})] 
\lrstep{\rho}{\A'} u[q] \lrstep{*}{\A'} q^\final \]
where $q, q_1,\ldots, q_m, q''_1,\ldots, q''_n ,q'_1,\ldots, q'_{m'} \in Q_L$ and $q^\final \in Q^\final_L$.
The transition rule denoted $\rho$ in the above sequence has the form $b(B) \to q$ and
$q_1\ldots q_m\, q''_1,\ldots, q''_n\, q'_1\ldots q'_{m'}$ is recognized by $B$, 
with a sequence 
$\init_B \lrstep{q_1\ldots q_m}{B} s \lrstep{q''_1\ldots q''_n}{B} s' 
 \lrstep{q'_1\ldots q'_{m'}}{B} f_B$,
 where $s, s'$ are two states of $B$.
By completeness of $\A_L$, given $s, s'$, there exists $q'$ such that
$a(B_{s,s'}) \hookrightarrow_{\Delta_i} q'$.
It follows in particular that $a(v) \lrstep{*}{\A'} q'$.
By construction, a transition $b(B + \langle s, q', s\rangle) \to q$ 
has been added to $\A'$, 
and $q_1\ldots q_m\, q'\, q'_1\ldots q'_{m'}$ is recognized by $B + \langle s, q', s\rangle$.
It follows that 
$u = u[b(h a(v)\ell)] \lrstep{*}{\A'} u[a(q_1\ldots q_m q' q'_1\ldots q'_{m'})] \lrstep{}{\A'} u[q] \lrstep{*}{\A'} q^\final$,
hence that $u \in L(\A')$.

\hfill(end Lemma direction $\subseteq$)\qed
\end{proof}

\hfill (end of the proof of Theorem~\ref{th:pre})
\end{RR}
\qed
\end{proof}

%
%

\section{Appendix: proof of Theorem~\ref{th:undec}}

\paragraph{\textsc{Theorem}~\ref{th:undec}.}
{\it
Reachability is undecidable for uniform PGTRS without variables and parameters.}

\medskip
\begin{proof}
We will reduce the halting problem of Deterministic Turing Machines (TM) 
that work on half a tape (unbounded on the right). 
We consider the following unary symbols to represent the tape  alphabet 
$\Sigma =\{0,1,\sharp,\flat\}$. 
We need a copy of the alphabet  $\Sigma' = \{0',1',\sharp',\flat'\}$.
We  only use  $\sharp$ to mark  the left endpoint  of the tape  
and $\flat$  is the  blank symbol,  e.g. representing the rightmost part of the tape.

The state symbols  are  constants in a finite set $ Q \cup Q'$ where $Q =\{q_1, q_2, \ldots , q_n\}$ 
and  $Q' =\{q_1', q_2', \ldots  , q'_n\}$. Hence each state of the TM has two representations. 

In order to represent a Turing machine configuration as a ground term we shall 
introduce  a binary symbol $+$ and a nullary symbol $\bot$. 
Now the TM configuration with tape $abccde\flat\flat \ldots$,  symbol under head $d$, state $q$  will be represented by:
$$\sharp(\bot) + (a(\bot) + (b(\bot)+ (c(\bot) + (c(\bot)+ (d(q)+ (c(\bot) + \flat(\bot))))))).$$
We denote by $\T_0$ (resp. $\T_1$) the set of terms on signature $\Sigma \cup \{ \bot, +\}$ 
with no occurrence of $\sharp$ (resp. with a unique occurrence of $\sharp$ at position 1).
Given a term in $t\in \T_0$ and a term $s\in \T(\Sigma)$ we write $t[\bot \leftarrow s]$  the term 
obtained from $t$ by replacing its 
rightmost  $\bot$ symbol by $s$. 

For each TM transition we introduce some rewrite rules that simulate it on the term representation. 
We introduce now some tree regular languages: 
$L_{s,a}$ is the subset of $t \in \T(\Sigma)$ such that $t$ admits a single occurrence of 
a state symbol and this state symbol is $s$, and it occurs right below a symbol  $a$.

\noindent{"In state $q$ reading $a$ go to state $r$ and write $b$".} 
This is translated to the ground rewrite rule:
$$ L_{q,a}  :: a(q) \rightarrow b(r)$$ 

\noindent{"In state $q$ reading $a$ go to state $r$ and move right".} 
This can be simulated by some  application of rules:
\begin{eqnarray}
L_{q,a}  :: u(\bot)     &\rightarrow& u(r')  \mbox{ ~for all ~}  u \in \{0,1,\sharp\} \\
L_{q,a}  :: \flat(\bot) &\rightarrow &\flat(r') + \flat(\bot) 
\end{eqnarray}
Note that one of these rule application may create a pattern $a(q) + (b(r')+ x)$  at the location where we had a pattern 
$a(q) + (b(\bot)+ x)$ in the configuration. 
Let $L_{q,a,r,R}$ be the set of term of type 
$U[\bot \leftarrow (a(q) + (b(r')+ V)]$ where $U \in \T_1$, $V \in \T_0$.
This is clearly a regular language. Then we add the rules: 
\begin{eqnarray}
 L_{q,a,r,R} :: a(q) &\rightarrow &a(\bot) \\ 
 L_{r',u} :: u(r')   &\rightarrow & u(r) \mbox{ ~for all ~}  u \in \{0,1,\sharp\} 
\end{eqnarray}

\noindent{"In state $q$ reading $a$ go to state $r$ and move left".} 
This can be simulated by some  application of rules:
\begin{eqnarray}
L_{q,a}  :: u(\bot)     &\rightarrow& u(r')  \mbox{ ~for all ~}  u \in \{0,1,\sharp\}
\end{eqnarray}
This rule application may create a pattern $b(r') + (a(q)+ x)$  at the location where we had a pattern 
$b(\bot) + (a(q)+ x)$ in the configuration. 
Let $L_{q,a,r,L}$ be the set of term of type $U[\bot \leftarrow ((b(r')+ a(q)) +V)]$ where $U\in \T_1$, 
$V \in \T_0$.
This is clearly a regular language. Then we add the rules: 
\begin{eqnarray}
 L_{q,a,r,L} :: a(q) &\rightarrow &a(\bot) \\
 L_{r',u} :: u(r')   &\rightarrow & u(r) \mbox{ ~for all ~}  u \in \{0,1,\sharp\} 
\end{eqnarray}
Let us denote $
\R = \{ L_i::\ell_i \rightarrow r_i \mid 1 \leq i \leq n \}$ 
the set of rules we obtain by the above construction.
Note that the languages $L_i$ are pairwise disjoint. 
By case inspection we can show that for any couple of TM configurations $T_1, T_2$
and their respective term encodings $t_1, t_2$, there is a sequence of transitions from 
$T_1$ to $T_2$ iff $t_1 \lrstep{*}{\R} t_2$. 
If we replace in every rule the regular language $L_i$  
by the disjoint union  $\biguplus_{1 \leq i \leq n} L_i$, the result still holds.
The theorem follows. 
\qed
\end{proof}

%
%

\section{Appendix: proof of Theorem~\protect\ref{th:pre2}}

\paragraph{\textsc{Theorem}~\ref{th:pre2}.}
{\it
Given a HA $\A$ on $\Sigma$ and a PTRS $\ptrs{\R}{\A} \in \XACU_2+$,
for all HA language $L$, $\pre^*_{\ptrs{\R}{\A}}(L)$
is a HA the language.}

\medskip
\begin{proof}
The proof is very close to the one of Theorem~\ref{th:pre}.
Indeed, in the above construction for Theorem~\ref{th:pre},
we consider the applications of rules 
$\INS_\mathsf{left}$, $\INS_\mathsf{right}$, $\RPL'$,
$\DEL$ and  $\DEL_{\mathsf{s}}$ under any symbol $b \in \F$.
Here instead, we can restrict the construction to the application under
the symbol specified in the lhs of the rewrite rules.
More precisely, let us just detail below the cases of the construction which are modified.
The rest of the prof is the same as for Theorem~\ref{th:pre}.

\begin{description}
\item{$\INS_{2,\mathsf{left}}$:} 
if $b(y\,a(x)\,z) \to b(y\,p\, a(x)\,z) \in \ptrs{\R}{\A}$, 
$B, B' \in \C$, $s, s'$ are states of $B$, and
$q, q_p ,q' \in Q_L$ 
such that 
$b(B) \to q \in \Delta_i$,
$a(B') \hookrightarrow_{\Delta_i} q'$,
$L(\A_i, q_p) \cap L(\A, p) \neq \emptyset$,
$s \lrstep{q_p q'}{B} s'$,
then 
$\Delta_{i+1} := \Delta_i \cup \bigl\{ b(B + \langle s, q', s'\rangle) \to q \bigr\}$.

\item{$\INS_{2,\mathsf{right}}$:} 
if $b(y\,a(x)\,z) \to b(y\, a(x)\,p\,z) \in \ptrs{\R}{\A}$, 
$B, B' \in \C$, $s, s'$ are states of $B$, and
$q, q_p ,q' \in Q_L$ 
such that 
$b(B) \to q \in \Delta_i$,
$a(B') \hookrightarrow_{\Delta_i} q'$,
$L(\A_i, q_p) \cap L(\A, p) \neq \emptyset$,
$s \lrstep{q' q_p}{B} s'$,
then 
$\Delta_{i+1} := \Delta_i \cup\bigl\{ b(B + \langle s, q', s'\rangle) \to q \bigr\}$.

\item{$\RPL'_2$:} 
if $b(y\,a(x)\,z) \to b(y\,p_1\ldots p_n\,z) \in \ptrs{\R}{\A}$, 
$B, B' \in \C$, $s, s'$ are states of $B$, and
$q, q', q_1,\ldots, q_n \in Q_L$ 
such that 
$b(B) \to q \in \Delta_i$,
$a(B') \hookrightarrow_{\Delta_i} q'$,
$L(\A_i, q_j) \cap L(\A, p_j) \neq \emptyset$ for all $1 \leq j \leq n$,
$s \lrstep{q_1\ldots q_n}{B} s'$
then 
$\Delta_{i+1} := \Delta_i \cup \bigl\{ b(B + \langle s, q', s'\rangle) ) \to q \bigr\}$.

\item{$\DEL_2$:} 
if $b(y\,a(x)\,z) \to b(yz) \in \ptrs{\R}{\A}$,
$B, B' \in \C$, $s$ is a state of $B$,
$q, q' \in Q_L$ 
such that 
$b(B) \to q \in \Delta_i$,
$a(B') \hookrightarrow_{\Delta_i} q'$,
then 
$\Delta_{i+1} := \Delta_i \cup \bigl\{ b(B + \langle s,q',s\rangle) \to q \bigr\}$.

\item{$\DEL_{2,\mathsf{s}}$:}
if $b(y\,a(x)\,z) \to b(yxz) \in \ptrs{\R}{\A}$,
$B \in \C$, $s, s'$ are states of $B$,
$q, q' \in Q_L$ 
such that 
$b(B) \to q \in \Delta_i$,
$a(B_{s,s'}) \hookrightarrow_{\Delta_i} q'$,
then 
$\Delta_{i+1} := \Delta_i \cup \bigl\{ b(B + \langle s,q',s'\rangle) \to q \bigr\}$.\qed
\end{description}
\end{proof}

%
%
%




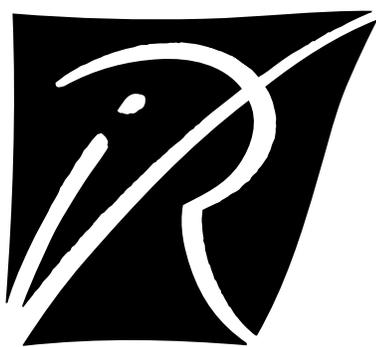
\end{document}